\documentclass[onecolumn,draft,11pt]{IEEEtran}
\ifCLASSINFOpdf
\else
\fi
\newcommand{\GRS}{{\mathrm{GRS}}}
\newcommand{\Hull}{{\mathrm{Hull}}}
\newcommand{\rank}{{\mathrm{rank}}}
\newcommand{\C}{{\mathcal{C}}}
\newcommand{\F}{{\mathbb{F}}}

\newcommand{\tabincell}[2]{\begin{tabular}{@{}#1@{}}#2\end{tabular}}

\makeatletter

\newcommand{\Rmnum}[1]{\expandafter\@slowromancap\romannumeral #1@}
\makeatother

\usepackage{lineno,hyperref,mathtools}
\usepackage[letterpaper,top=1.5cm,bottom=1.5cm,left=2cm,right=2cm,marginparwidth=1.75cm]{geometry}
\usepackage{float}
\usepackage{threeparttable}
\usepackage{amssymb}
\usepackage{bm}
\usepackage{amssymb}
\usepackage{amsthm}
\usepackage{multirow,booktabs}
\usepackage{cases}
\usepackage{tabularx}
\usepackage{adjustbox}
\usepackage[figuresright]{rotating}
\usepackage{longtable}
\usepackage{booktabs}
\usepackage{color}
\usepackage{ulem}
\usepackage{setspace}

\newtheorem{theorem}{Theorem}

\newtheorem{remark}{Remark}
\newtheorem{definition}[theorem]{Definition}
\newtheorem{lemma}[theorem]{Lemma}
\newtheorem{corollary}[theorem]{Corollary}
\newtheorem{example}[theorem]{Example}

\begin{document}
\begin{sloppypar}

%
\title{On Galois hulls of linear codes and new entanglement-assisted quantum error-correcting codes \thanks{First and second authors were supported by the National Natural Science Foundation of China 
		(Nos.U21A20428 and 12171134). }}
%
%
%

\author{Yang~Li, Shixin~Zhu$^\dag$\thanks{$^\dag$ Corresponding author}
\thanks{Yang~Li and  Shixin~Zhu are with the School of Mathematics, Hefei University of Technology, Hefei 230601, China (e-mail: yanglimath@163.com, zhushixinmath@hfut.edu.cn).}
\thanks{Manuscript received --; revised --}}

%
%

\markboth{Journal of \LaTeX\ Class Files}%
{Shell \MakeLowercase{\textit{et al.}}: Bare Demo of IEEEtran.cls for Journals}
%



\maketitle

\begin{abstract}
    The Galois hull of a linear code is the intersection of itself and its Galois dual code, which has aroused 
    the interest of researchers in these years. In this paper, we study Galois hulls of linear codes. 
    Firstly, the symmetry of the dimensions of Galois hulls of linear codes is found. 
    Some new necessary and sufficient conditions for linear codes being Galois self-orthogonal codes, 
    Galois self-dual codes, and Galois linear complementary dual codes are characterized. 
    Then, we propose explicit methods to construct Galois self-orthogonal codes 
    of larger length from given Galois self-orthogonal codes.
    As an application, linear codes of larger length with Galois hulls of arbitrary dimensions are further derived. 
    Focusing on the Hermitian inner product, two new classes of Hermitian self-orthogonal maximum distance separable (MDS) codes are also constructed. 
    Finally, applying all the results to the construction of entanglement-assisted quantum error-correcting codes (EAQECCs), 
    many new $q$-ary or $\sqrt{q}$-ary EAQECCs and MDS EAQECCs with rates greater than or equal to $\frac{1}{2}$ and positive net rates can be obtained. 
    Moreover, the minimum distance of many $\sqrt{q}$-ary MDS EAQECCs of length $n>\sqrt{q}+1$ is greater than or equal to $\lceil \frac{\sqrt{q}}{2} \rceil$.  
\end{abstract}


%
\IEEEpeerreviewmaketitle

\section{Introduction}\label{sec-introduction}
\subsection{Quantum codes and hulls of linear codes}
Quantum error-correcting codes (QECCs) can correct qubit errors caused by noise in quantum communication and 
quantum computation \cite{RefJ-16,RefJ7,RefJ-17,RefJ-15} and have important applications in entanglement 
distillation protocols \cite{RefJ-18}. In general, it is believed that constructions of good QECCs are much 
more difficult than good classical error-correcting codes. In 1996, Calderbank, Shor and Steane 
\cite{RefJ6,RefJ7} proposed a systematic method, namely CSS construction, to construct stabilizer QECCs 
from linear codes with certain self-orthogonality, which undoubtedly and greatly facilitated the development of QECCs.  
However, self-orthogonal linear codes are usually not readily available. Later, Burn et al. \cite{RefJ10} 
introduced the so-called entanglement-assisted quantum error-correcting codes (EAQECCs), which contain stabilizer 
QECCs as a special case. 

Throughout this paper, $p$ is a prime and $h$, $e$ are two positive integers satisfying $0\leq e\leq h-1$. Let $q=p^h$.  
We use $[[n,k,d;c]]_q$ to denote a $q$-ary EAQECC that encodes $k$ information qubits into $n$ channel qubits with the help of $c$ 
ebits, and $d$ is called the minimum distance of the EAQECC. 
For such an EAQECC, its performance is measured by its rate $\rho:=\frac{k}{n}$ and net rate $\bar{\rho}:=\frac{k-c}{n}$. 
Notably, Brun et al. \cite{RefJ B6} proved that an EAQECC is possible to be used to construct catalytic codes if $\bar{\rho}>0$. 
These interesting findings demonstrated that as long as we can determine the number of $c$, then EAQECCs will be easily derived 
from arbitrary linear codes. 
Fortunately, this is entirely possible since the close relationships between the number of $c$ and the dimension of the hull of 
a linear code $\mathcal{C}$ had been found in \cite{RefJ12,RefJ13}. 

Let $\mathcal{C}$ be a linear code and $\mathcal{C}^{\bot}$ be the dual code of $\mathcal{C}$ with respect to certain 
inner product. Then the hull of $\mathcal{C}$ mentioned above is defined by the linear code $\mathcal{C}\cap \mathcal{C}^{\bot}$. 
Generally, we denote it by $\Hull(\mathcal{C})=\mathcal{C}\cap \mathcal{C}^{\bot}$. At first, the hull of $\mathcal{C}$ 
was introduced to classify finite projective planes by Assmus et al. in \cite{RefJ11}. Then, some celebrated studies 
have shown that the hull of linear codes has many important applications in coding theory, such as determining the complexity 
of algorithms for computing the automorphism group of a linear code \cite{RefJ16} and for checking the 
permutation equivalence of two linear codes \cite{RefJ B1,RefJ19}, as well as calculating the number of $c$ in an EAQECC \cite{RefJ12,RefJ13}. 

\subsection{Related works}
On the basis of these important application scenarios, there has been great enthusiasm in recent years for the study of properties 
of hulls of linear codes. 
Concretely, the Euclidean, Hermitian or general Galois hulls of cyclic codes, negacyclic codes, constacyclic codes, 
Bose-Chaudhuri-Hocquenghem (BCH) codes and (extended) generalized Reed-Solomon (GRS) codes were deeply studied in 
\cite{RefJ90,RefJ91,RefJ-5,RefJ94,RefJ24,RefJ92,RefJ95} and the references therein. 
Moreover, as applications, some self-orthogonal codes, self-dual codes, linear complementary dual (LCD) codes, EAQECCs and 
maximum distance separable (MDS) EAQECCs were constructed. 
In particular, some hull-invariant problems were discussed in \cite{RefJ93,RefJ24}. 

In all these works, there is a particular interest in constructing hulls of prescribed dimensions (usually small dimensions) and in 
constructing hulls of arbitrary dimensions. Very recently, Fang et al. \cite{RefJ26} and Li et al. \cite{RefJ27} constructed 
MDS codes with Galois hulls of arbitrary dimensions from Euclidean, Hermitian and general Galois self-orthogonal (extended) GRS codes. 
Sok \cite{RefJ5} proved that linear codes of length $n$ with $1$-dimensional Euclidean hull can be derived from Euclidean 
self-orthogonal codes of length $n$. He also explained that linear codes with Euclidean hulls of arbitrary dimensions can actually 
be obtained. 
Sok \cite{RefJ48}, Chen \cite{RefJ47} and Luo et al. \cite{RefJ B4} independently presented a similar conclusion  
under the Hermitian inner product. Particularly, in \cite{RefJ47}, Chen also extended the property to general Galois self-orthogonal linear 
codes. Again, as a consequence, many new EAQECCs and MDS EAQECCs were instantly obtained. 


\subsection{Our contributions}
Inspired and motivated by these works, especially the works on constructions of linear codes with hulls of arbitrary or prescribed 
dimensions from self-orthogonal codes, we study Galois hulls of linear codes in this paper. 
\begin{enumerate}
    \item [\rm (1)] The properties of the dimensions of Galois hulls of linear codes are studied. 
    \begin{itemize}
        \item 
        We give a new formula for counting the dimensions of Galois hulls of linear codes. 
        
        \item By using the new formula, we derive two important equations for the dimensions of Galois hulls of linear codes. 
        We note that these two equations have also been proved in the recent article \cite{RefJ96} and preprint \cite{RefJ-11}. 
        Compared with them, a new technical route is proposed to derive similar results.   

        \item Based on these properties, we show the equivalence of Galois self-orthogonal codes, Galois self-dual codes, 
        and Galois LCD codes under different Galois inner products. As an application, some new criteria for 
        determining whether a linear code is a Galois self-orthogonal code, a Galois self-dual code, or a Galois LCD code 
        are given. 
    \end{itemize}

    \item [\rm (2)] We construct linear codes of larger length with Galois hulls of arbitrary dimensions from given Galois self-orthogonal codes. 
    \begin{itemize}
        \item Under certain relatively weak conditions, we give a deterministic method to increase the length of Galois 
        self-orthogonal codes while maintaining self-orthogonality and possibly boosting the minimum distance.

        \item Combining the work of Chen \cite{RefJ47}, we further show that linear codes of larger length with Galois hulls of arbitrary dimensions can be constructed 
        from known Galois self-orthogonal codes when $q\geq 3$.  
    \end{itemize} 

    \item [\rm (3)] Focusing on the Hermitian case, we construct two new classes of Hermitian self-orthogonal MDS codes. 
    These new Hermitian self-orthogonal MDS codes can fully cover some known results in previous literature. 

    \item [\rm (4)] We apply all results to construct new $q$-ary and $\sqrt{q}$-ary EAQECCs and MDS EAQECCs.  
    As a consequence, some EAQECCs with rates greater than or equal to $\frac{1}{2}$ and positive net rates can be obtained. 
    Moreover, many $\sqrt{q}$-ary MDS EAQECCs of length $n>\sqrt{q}+1$ have the minimum distance greater than or equal to 
    $\lceil \frac{\sqrt{q}}{2} \rceil$. 

\end{enumerate} 

\subsection{Organization of this paper}
The organization of this paper is presented as follows. 
In Section \ref{sec2}, some basic notations and results on Galois hulls and GRS codes are reviewed. 
In Section \ref{sec3}, the properties of dimensions of Galois hulls of linear codes are studied. 
In Section \ref{sec4.1}, the explicit constructions of Galois self-orthogonal codes of larger length 
from known Galois self-orthogonal codes under certain conditions are given. 
In Section \ref{sec4.2}, the conditions required are proved to be relatively weak. 
In Section \ref{sec4.3}, these conditions are further refined under the usual Euclidean case and the classical Hermitian case. 
In Section \ref{sec4.4}, the existence of linear codes of larger length with Galois hulls of arbitrary dimensions is shown. 
In Section \ref{sec4.5-Examples}, many interesting examples on Galois self-orthogonal codes are listed. 
In Section \ref{sec5}, two new constructions of Hermitian self-orthogonal MDS codes are presented. 
In Section \ref{sec6}, based on conclusions obtained in previous sections, new constructions of EAQECCs and MDS EAQECCs are proposed.   
Finally, in Section \ref{sec7}, some concluding remarks are discussed.

\section{Preliminaries}\label{sec2}

\subsection{Some basic notations}\label{sec2.1}
From this section onwards, we fix the following notations, unless it is stated otherwise:
\begin{itemize}
    \item 
    \begin{itemize} 
        \item $O_{l\times l}$ is a zero matrix of size $l\times l$; 
        \item $\mathbf{1}_{k\times k}$ is a matrix of size $k\times k$ whose elements are all $1$;  
        \item $I_k$ is an identity matrix of order $k$. 
    \end{itemize}
    \item $\rank(G)$ is the rank of a matrix $G$.
    \item $\C$ is a linear code with length $n$, dimension $k$ and minimum Hamming distance $d$ over $\F_q$, 
    denoted by $[n,k,d]_{q}$. Sometimes we omit the words ``$\emph{linear}$'' and ``$\emph{Hamming}$''.
    \item $\omega_H(\mathbf{c})$ is the Hamming weight of a codeword $\mathbf{c}\in \C$, which is defined by the 
    number of nonzero components of $\mathbf{c}$. 
    \item 
    \begin{itemize} 
        \item $\dim(\C)=k$ is the dimension of $\C$;    
        \item $d_H(\C)=d$ is the minimum Hamming distance of $\C$, which equals the minimum weight of all nonzero codewords of $\C$.  
    \end{itemize}

    \item 
    \begin{itemize} 
        \item $\lfloor \cdot \rfloor$ is the floor function;     
        \item $\lceil \cdot \rceil$ is the ceil function. 
    \end{itemize}
\end{itemize}

\subsection{Galois hulls}\label{sec2.2}
Let $\mathbb{F}_{q}$ be the finite field with size $q$. Then $\F_{q}^n$ is an $n$-dimensional vector space over $\F_q$ 
and an $[n,k,d]_q$ code $\C$ can be seen as a $k$-dimensional subspace of $\F_q^n$. 
For any two vectors $\mathbf{x}=(x_1,x_2,\dots,x_n)$, $\mathbf{y}=(y_1,y_2,\dots,y_n)\in \mathbb{F}_{q}^n$, we can define the $e$-Galois 
inner product of $\mathbf{x}$ and $\mathbf{y}$ as  
\begin{equation}
    \langle \mathbf{x},\mathbf{y} \rangle_e = \sum_{i=1}^{n}x_iy_i^{p^e},\ {\rm {where}}\ 0\leq e\leq h-1.
\end{equation}
The $e$-Galois inner product was first introduced by Fan et al. \cite{RefJ54} and is a generalization of the usual Euclidean inner 
(i.e., $e=0$) and the classical Hermitian inner product (i.e., $e=\frac{h}{2}$ and $h$ is even). The code 
\begin{equation}
    \mathcal{C}^{\bot_e}=\{\mathbf{x}\in \mathbb{F}_{q}^n:\ \langle \mathbf{x},\mathbf{y} \rangle_e=0\ {\rm {for\ all}}\ \mathbf{y}\in \mathcal{C}\} 
\end{equation}
is always called the $e$-Galois dual code of $\mathcal{C}$. Then $\C^{\bot_0}$ (resp. $\C^{\bot_\frac{h}{2}}$ with even $h$) is just 
the Euclidean (resp. Hermitian) dual code of $\C$. 
Denote the $e$-Galois hull of $\C$ by $\Hull_e(\mathcal{C})=\mathcal{C}\cap \mathcal{C}^{\bot_e}$. 
Similarly, $\Hull_0(\C)$ (resp. $\Hull_{\frac{h}{2}}(\C)$ with even $h$) coincides with the Euclidean (resp. Hermitian) hull of $\C$. 
Conversely, in this paper, when we talk about the Hermitian inner product or the Hermitian hull, the $h$ in $q=p^h$ is even. 
Let $\mathbf{0}$ be the zero vector. The length of $\mathbf{0}$ is unspecified here and depends on the context. Then $\mathcal{C}$ is 
an $e$-Galois LCD code if $\Hull_e(\mathcal{C})=\{\mathbf{0}\}$, an $e$-Galois self-orthogonal code if $\Hull_e(\mathcal{C})=\mathcal{C}$, 
and an $e$-Galois self-dual code if $\Hull_e(\mathcal{C})=\mathcal{C}=\mathcal{C}^{\bot_e}$. 

Let 
\begin{align}
    \begin{split}
        \sigma:\  \mathbb{F}_q & \longrightarrow  \mathbb{F}_q, \\ 
                a & \longmapsto  a^p,
    \end{split}
\end{align}
be the Frobenius automorphism of $\mathbb{F}_q$. For any vector 
$\mathbf{x}=(x_1,x_2,\dots,x_n) \in \mathbb{F}_{q}^n$ and any matrix $G=(m_{ij})_{s\times t}$ over $\mathbb{F}_q$, we denote 
$$\sigma(\mathbf{x})=(\sigma(x_1),\sigma(x_2),\dots,\sigma(x_n))\ {\rm and}\ \sigma(G)=(\sigma(m_{ij}))_{s\times t}.$$ 
Hence, $\sigma^{h-e}(G^T)=(\sigma^{h-e}(G))^T$, which is also often written as $G^\ddagger$ 
in many other literature.  

There are some useful results on Galois dual codes and Galois hulls. We rephrase them as follows.

\begin{lemma}\label{lemma3}{\rm (\cite{RefJ24}, \cite{RefJ55})}
    Let $\mathcal{C}$ be an $[n,k,d]_q$ code with a generator matrix $G$. Then $\sigma^{h-e}(\mathcal{C})$ is an $[n,k,d]_q$ 
    linear code with a generator matrix $\sigma^{h-e}(G)$ and 
    $\mathcal{C}^{\bot_e}=(\sigma^{h-e}(\mathcal{C}))^{\bot_0}=\sigma^{h-e}(\mathcal{C}^{\bot_0})$. 
    Moreover, $(\mathcal{C}^{\bot_e})^{\bot_{h-e}}=\mathcal{C}$.
\end{lemma}

\begin{lemma}{\rm (\cite{RefJ13})}\label{lemma.the ralation between hull and rank(HH)}
    Let $\mathcal{C}$ be an $[n,k,d]_q$ code with a parity check matrix $H$. Then  
    \begin{align}\label{eq.the ralation between hull and rank(HH)}
        \rank(H\sigma^{h-e}(H^T))=n-k-\dim(\Hull_e(\mathcal{C})).
    \end{align}
\end{lemma}

\begin{lemma}\label{lem.the relationship between hull and generator matrix}{\rm (\cite{RefJ24})}
    Let $\mathcal{C}$ be an $[n,k,d]_q$ code with a generator matrix $G$. 
    Let $\dim(\Hull_e(\C))=l$ and $r=k-l$. Then there exists a generator matrix $G_0$ of $\mathcal{C}$ 
    such that 
    \begin{align*}
        G_0\sigma^{e}(G_0^{T})=\begin{pmatrix*}
            O_{l\times l} & H_{l\times r}\\
            O_{r\times l} & P_{r\times r}\\
        \end{pmatrix*},
    \end{align*}
    where $Q=\begin{pmatrix*}
        H_{l\times r}\\
        P_{r\times r}\\
    \end{pmatrix*}$ is a matrix of size $(l+r)\times r$ and $\rank(Q)=r$. 
    Moreover, $\rank(G\sigma^{e}(G^{T}))=r$ for any generator matrix $G$ of $\mathcal{C}$. 
\end{lemma}
\begin{remark}\label{rem1.the relationship between hull and generator matrix}
    In a more concise language, Lemma \ref{lem.the relationship between hull and generator matrix} implies that for an 
    $[n,k,d]_q$ code $\C$, $\rank(G\sigma^{e}(G^{T}))$ 
    is independent of $G$ such that 
    \begin{align}\label{eq.panxu hull}
        \rank(G\sigma^{e}(G^{T}))=k-\dim(\Hull_e(\mathcal{C})), 
    \end{align}
    where $k=\dim(\C)$ and $G$ is a generator matrix of $\C$.   
\end{remark}

Additionally, starting from Galois self-orthogonal codes, one can easily obtain linear codes with Galois hulls of arbitrary dimensions 
from the following lemmas. 
\begin{lemma}\label{coro.all_Galois hull from Galois self-orthogonal code}{\rm (\cite{RefJ47})}
    Let $q\geq 3$ be a prime power. If there exists an $[n,k,d]_q$ $e$-Galois self-orthogonal code, then there exists 
    an $[n,k,d]_q$ code with $l$-dimensional $e$-Galois hull for each $0\leq l\leq k$.
\end{lemma}

\subsection{Generalized Reed-Solomon (GRS) codes}\label{sec2.3}
For an $[n,k,d]_q$ code $\C$, there exists a so-called Singleton bound, which yields $d\leq n-k+1$. 
Particularly, if $d=n-k+1$, then $\C$ is called an MDS code. Let $\F_q^*=\F_q \setminus \{0\}$.  
As special MDS codes, we now review some basic knowledge about 
generalized Reed-Solomon (GRS) codes. Let $\mathbf{a}=(a_1,a_2,\dots,a_n)\in \mathbb{F}_q^n$ with $a_i\neq a_j$ for any 
$1< i\neq j\leq n\leq q$ and $\mathbf{v}=(v_1,v_2,\dots,v_n)\in (\mathbb{F}_q^*)^n$. For an integer $1\leq k\leq n$, 
the $k$-dimensional generalized Reed-Solomon (GRS) code of length $n$ over $\F_q$ associated to $\mathbf{a}$ and $\mathbf{v}$ 
is defined by  
\begin{align}
  \GRS_k(\mathbf{a},\mathbf{v})=\{(v_1f(a_1),v_2f(a_2),\dots,v_nf(a_n)):\ f(x)\in \mathbb{F}_q[x],\ \deg(f(x))\leq k-1\}.
\end{align}
A generator matrix of $\GRS_k(\mathbf{a},\mathbf{v})$ is given by 
\begin{align}
    G_k(\mathbf{a},\mathbf{v}) = \left( 
                \begin{array}{cccc}
                v_1 & v_2  &\cdots & v_n \\
                v_1a_1 & v_2a_2 & \cdots & v_na_n \\
                v_1a_1^2 & v_2a_2^2 & \cdots & v_na_n^2 \\
                \vdots & \vdots &\ddots & \vdots \\
                v_na_1^{k-1} & v_na_2^{k-1} & \cdots & v_na_n^{k-1} 
                \end{array}
    \right).       
\end{align}

The elements $a_1, a_2, \cdots, a_n$ are called the $\underline{code\ locators}$ of $\GRS_k(\mathbf{a}, \mathbf{v})$, 
and the elements $v_1, v_2, \cdots, v_n$ are called the $\underline{column\ multipliers}$ of $\GRS_k(\mathbf{a}, \mathbf{v})$. 
From now on, for each $1\leq i\leq n$, we set
\begin{align}\label{equation_ui}
    u_i=\prod_{1\leq j\leq n,i\neq j}(a_i-a_j)^{-1},
\end{align}
which can be used to help determine whether a codeword of a GRS code belongs to its Hermitian dual code. 
We write an alternative criterion introduced by Guo et al. \cite{RefJ-4} as follows. 

\begin{lemma}{\rm (\cite{RefJ-4})}\label{lemma.GUO___Hermitian self-orthogonal GRS}
    A codeword $\mathbf{c}=(v_1f(a_1),v_2f(a_2),\dots,v_nf(a_n))$ of $\GRS_k(\mathbf{a},\mathbf{v})$ is contained in 
    $\GRS_k(\mathbf{a},\mathbf{v})^{\bot_{\frac{h}{2}}}$ if and only if there exists a monic polynomial 
    $h(x)\in \mathbb{F}_q[x]$ with $\deg(h(x)f^{\sqrt{q}}(x))\leq n-k-1$ such that 
    \begin{align}
        \lambda u_ih(a_i)=v_i^{\sqrt{q}+1},\ 1\leq i\leq n,
    \end{align}
    where $\lambda\in \mathbb{F}_q^*$ and $\sqrt[]{q}=p^{\frac{h}{2}}$.
\end{lemma}

\section{Symmetry of dimensions of Galois hulls of codes}\label{sec3}
In this section, we first give a new formula to calculate $\dim(\Hull_e(\mathcal{C}))$. Then combining Equation (\ref{eq.panxu hull}), 
we further show that the dimensions of Galois hulls of codes have symmetry. Then, the equivalence of Galois self-orthogonal codes, 
Galois self-dual codes, and Galois LCD codes under different Galois inner products is demonstrated. Finally, we give some new 
criteria of judging whether a code is a Galois self-orthogonal code, a Galois self-dual code, or a Galois LCD code.  

\begin{lemma}\label{lemma.the ralation between hull and rank(GG)}
    Let $\mathcal{C}$ be an $[n,k,d]_q$ code with a generator matrix $G$. Then $\rank(G\sigma^{h-e}(G^T))$ is independent of $G$ such that 
    \begin{align}\label{eq.hull new formula}
        \rank(G\sigma^{h-e}(G^T))=k-\dim(\Hull_e(\mathcal{C})).
    \end{align}
\end{lemma}
\begin{proof}
    From Lemma \ref{lemma3}, we have $\mathcal{C}^{\bot_e}=(\sigma^{h-e}(\mathcal{C}))^{\bot_0}$ and $\sigma^{h-e}(G)$ is a 
    generator matrix of $\sigma^{h-e}(\mathcal{C})$. Hence, $\sigma^{h-e}(G)$ is a parity check matrix of $\mathcal{C}^{\bot_e}$. 
    From Lemma \ref{lemma.the ralation between hull and rank(HH)}, considering the $(h-e)$-Galois hull of $\C^{\bot_e}$, we have  
    \begin{align}\label{eq.the ralation between hull and rank(GG).1}
        \begin{split}
            \rank(\sigma^{h-e}(G)\sigma^e((\sigma^{h-e}(G))^T)) & = n - \dim(\C^{\bot_e}) - \dim(\Hull_{h-e}(\C^{\bot_e})) \\
                                                                & = n-(n-k)-\dim(\mathcal{C}^{\bot_e}\cap (\mathcal{C}^{\bot_e})^{\bot_{h-e}})\\ 
                                                                & = k-\dim(\mathcal{C}\cap \mathcal{C}^{\bot_e}) \\
                                                                & = k-\dim(\Hull_e(\C)). 
        \end{split}
    \end{align}
    It follows $\rank(\sigma^{h-e}(G)G^T)=k-\dim(\Hull_e(\mathcal{C}))$ from the fact that $\sigma^e((\sigma^{h-e}(G))^T)=G^T$. 
    And since $$\sigma^{h-e}(G)G^T=(G(\sigma^{h-e}(G))^T)^T=(G\sigma^{h-e}(G^T))^T,$$ 
    we get $\rank(\sigma^{h-e}(G)G^T)=\rank(G\sigma^{h-e}(G^T))$. 
    Therefore, the desired result 
    is proven. 
    
    Next, we need to prove that $\rank(G\sigma^{h-e}(G^T))$ is independent of $G$. Let $G_1$ be another generator matrix of $\mathcal{C}$. 
    After a suitable sequence of elementary row operations, we have $G=BG_1$ for some invertible $n\times n$ matrix $B$ over $\mathbb{F}_q$. 
    Since $\sigma$ is an automorphism of $\F_q$, we have  
    $$\rank(G\sigma^{h-e}(G^T))=\rank(BG_1\sigma^{h-e}(G_1^T)\sigma^{h-e}(B^T))=\rank(G_1\sigma^{h-e}(G_1^T)),$$ 
    which completes the proof. 
\end{proof}
\begin{remark}\label{rem2.lem8}

    Coming from a different motivation, in {\rm \cite[Lemma 3.2]{RefJ60}}, Liu et al. derived 
$$\rank(G_1(\sigma^{h-e}(G_2))^T)=k_1-\dim(\C_1\cap \C_2^{\bot_e}),$$ where $G_i$ is a generator matrix of $\C_i$ and $k_1=\dim(\C_1)$. 
In particular, taking $\mathcal{C}_1=\mathcal{C}_2=\mathcal{C}$, one can get 
$$\rank(G_1(\sigma^{h-e}(G_2))^T)=k-\dim(\Hull_e(\mathcal{C})),$$ where $G_1$, $G_2$ are two generator matrices of $\mathcal{C}$ and 
$k=\dim(\mathcal{C})$. Then after a similar calculation to Lemma \ref{lemma.the ralation between hull and rank(GG)}, a comparable 
result can be obtained. Therefore, Lemma \ref{lemma.the ralation between hull and rank(GG)} actually gives a different approach 
to the proof, and it will be helpful for some proofs in the following.     
\end{remark}

Based on Lemmas \ref{lem.the relationship between hull and generator matrix} and \ref{lemma.the ralation between hull and rank(GG)}, 
we can deduce the symmetry of dimensions of Galois hulls of codes as follows.

\begin{theorem}\label{th.symmetry of Galois hulls}
    Let $\mathcal{C}$ be an $[n,k,d]_q$ code. Then for any $0\leq e\leq h-1$, we have 
    \begin{align}\label{eq.symmetry of Galois hulls}
        \dim(\Hull_e(\mathcal{C}))=\dim(\Hull_{h-e}(\mathcal{C})).
    \end{align}
\end{theorem}
\begin{proof}
    The proof can be finished in two ways.  
    \begin{itemize}
        \item $\textbf{Case 1:}$ When $1\leq e\leq h-1$, we have $1\leq h-e\leq h-1$. 
        The desired result follows directly from Equations (\ref{eq.panxu hull}) and (\ref{eq.hull new formula}). 
        \item $\textbf{Case 2:}$ When $e=0$, we have $h-e=h$. We naturally define $\Hull_h(\mathcal{C})$ as 
        $$\Hull_h(\mathcal{C})=\mathcal{C}\cap \mathcal{C}^{\bot_h}=\mathcal{C}\cap (\sigma^{h-h}(\mathcal{C}))^{\bot_0}=\mathcal{C}\cap \mathcal{C}^{\bot_0}=\Hull_0(\C).$$ 
        Hence, for $e=0$, we also have $\dim(\Hull_0(\mathcal{C}))=\dim(\Hull_h(\mathcal{C}))$.
    \end{itemize}

 Combining $\textbf{Case 1}$ and $\textbf{Case 2}$, we complete the proof.
\end{proof}

\begin{remark}\label{remark3}  $\quad$
    \begin{enumerate}
        \item [\rm (1)] The natural definition of $\Hull_h(\mathcal{C})$ implies that we can also consider the $h$-Galois inner product 
        and treat it as the $0$-Galois inner product, which is helpful for us to deal with Galois hulls in a unified view. 
        \item [\rm (2)] Independently, over finite commutative chain rings, based on the tool of $R$-module epimorphism and relationships between the $q$-dimension 
        of a code $\C$ over Frobenius rings and the $q$-dimension of the $e$-Galois dual code $\C^{\bot_e}$, Talbi et al. \cite{RefJ96} gave a similar conclusion. 
        Note that Theorem \ref{th.symmetry of Galois hulls} is a direct application of Lemma \ref{lemma.the ralation between hull and rank(GG)} 
        and that only some elementary algebraic knowledge is used. 
    \end{enumerate}
\end{remark}

From Theorem \ref{th.symmetry of Galois hulls} and the proof of Lemma \ref{lemma.the ralation between hull and rank(GG)}, the relationship between 
$\dim(\Hull_e(\mathcal{C}))$ and $\dim(\Hull_e(\mathcal{C}^{\bot_e}))$ can be further deduced. 
Notably, this is an open problem left by Cao \cite{RefJ3}, and as Cao said, the relationship is important for the construction of EAQECCs 
from codes with prescribed dimensional Galois hull. 

\begin{corollary}
    \label{coro.C and C dual Hull equals}
    Let $\mathcal{C}$ be an $[n,k,d]_q$ code. Then for any $0\leq e\leq h-1$, we have   
    \begin{align}
        \dim(\Hull_e(\mathcal{C}))=\dim(\Hull_{e}(\mathcal{C}^{\bot_e})).
    \end{align}
\end{corollary}
\begin{proof}
    From the proof of Lemma \ref{lemma.the ralation between hull and rank(GG)}, we have 
    $\Hull_{h-e}(\mathcal{C}^{\bot_e})=\mathcal{C}^{\bot_e}\cap (\mathcal{C}^{\bot_e})^{\bot_{h-e}}=\Hull_e(\mathcal{C})$. 
    The desired result is then obtained instantaneously from Theorem \ref{th.symmetry of Galois hulls}. 
\end{proof}

\begin{remark}\label{rem2.Cao open problem} 

Very close to the time of our earlier preprint, Verma et al. \cite{RefJ-11} followed the approach introduced 
by Talbi et al. \cite{RefJ96} and the property $(\C_1+\C_2)^{\bot_e}=\C_1^{\bot_e}\cap \C_2^{\bot_e}$ to give 
a similar result, where $\C_1$ and $\C_2$ are two codes over $\F_q$. In Corollary \ref{coro.C and C dual Hull equals}, 
our conclusion follows from Lemma \ref{lemma.the ralation between hull and rank(GG)} and Theorem \ref{th.symmetry of Galois hulls}. 
According to Remark \ref{remark3} (2), these are obviously two different technical routes. 
We guess that by continuing the work in these two different routes, many other new consequences may be discovered.         

\end{remark}

Next, we consider three special cases of Galois hulls, namely, 
$e$-Galois self-orthogonal codes (i.e., $\Hull_e(\mathcal{C})=\mathcal{C}$),
$e$-Galois self-dual codes (i.e., $\Hull_e(\mathcal{C})=\mathcal{C}=\C^{\bot_e}$), and 
$e$-Galois LCD codes (i.e., $\Hull_e(\mathcal{C})=\{\mathbf{0}\}$). 
We further illustrate the equivalence of these three cases under the $e$-Galois inner 
product and the $(h-e)$-Galois inner product.

\begin{theorem}\label{th.self-orthogonal_equivence}
    Let $\mathcal{C}$ be an $[n,k,d]_q$ code. Then $\mathcal{C}$ is an 
    $e$-Galois self-orthogonal code if and only if $\mathcal{C}$ is an $(h-e)$-Galois self-orthogonal code.
\end{theorem}
\begin{proof}
    As Remark \ref{remark3} (1) states, the $h$-Galois inner product can be seen as the $0$-Galois inner product. 
    Then from the definition of Galois self-orthogonal code and Theorem \ref{th.symmetry of Galois hulls}, we know that for any $0\leq e\leq h-1$, 
    \begin{align*}
        \text{$\mathcal{C}$ is\ an\ $e$-Galois\ self-orthogonal\ code} & \Leftrightarrow \Hull_e(\mathcal{C})=\mathcal{C}\\
                                                 & \Leftrightarrow \dim(\Hull_e(\mathcal{C}))=\dim(\mathcal{C})=k\\
                                                 & \Leftrightarrow \dim(\Hull_{h-e}(\mathcal{C}))=\dim(\mathcal{C})=k\\
                                                 & \Leftrightarrow \Hull_{h-e}(\mathcal{C})=\mathcal{C}\\
                                                 & \Leftrightarrow \text{$\mathcal{C}$\ is\ an\ $(h-e)$-Galois\ self-orthogonal\ code},
    \end{align*}
    where the second and penultimate equivalence conditions hold due to the facts $\Hull_e(\mathcal{C})\subseteq \mathcal{C}$ and 
    $\Hull_{h-e}(\mathcal{C})\subseteq \mathcal{C}$. Therefore, we complete the proof. 
\end{proof}

As a special type of Galois self-orthogonal codes, the equivalence of $e$-Galois self-dual codes and $(h-e)$-Galois self-dual codes 
promptly follows from Theorem \ref{th.self-orthogonal_equivence} and the fact $\dim(\C^{\bot_e})=\dim(\C^{\bot_{h-e}})=n-k$. 
We write it in the next theorem and omit the proof. 

\begin{theorem}\label{th.self-dual_equivence}
    Let $\mathcal{C}$ be an $[n,k,d]_q$ code. Then $\mathcal{C}$ is an $e$-Galois self-dual code 
    if and only if $\mathcal{C}$ is an $(h-e)$-Galois self-dual code.
\end{theorem}

Similarly, we can also derive the equivalence of $e$-Galois LCD codes and $(h-e)$-Galois LCD codes. 
\begin{theorem}\label{th.LCD_equivence}
    Let $\mathcal{C}$ be an $[n,k,d]_q$ code. Then $\mathcal{C}$ is an 
    $e$-Galois LCD code if and only if $\mathcal{C}$ is an $(h-e) $-Galois LCD code.
\end{theorem}
\begin{proof}
    From the definition of Galois LCD code and Theorem \ref{th.symmetry of Galois hulls}, we know that for any $0\leq e\leq h-1$, 
    \begin{align*}
        \text{$\mathcal{C}$ is\ an\ $e$-Galois\ LCD\ code} & \Leftrightarrow \Hull_e(\mathcal{C})=\{\mathbf{0}_n\}\\
                                                 & \Leftrightarrow \dim(\Hull_e(\mathcal{C}))=0\\
                                                 & \Leftrightarrow \dim(\Hull_{h-e}(\mathcal{C}))=0\\
                                                 & \Leftrightarrow \Hull_{h-e}(\mathcal{C})=\{\mathbf{0}_n\}\\
                                                 & \Leftrightarrow \text{$\mathcal{C}$\ is\ an\ $(h-e)$-Galois\ LCD\ code},
    \end{align*}
    which completes the proof. 
\end{proof}

Moreover, according to Remark \ref{rem1.the relationship between hull and generator matrix}, Lemma \ref{lemma.the ralation between hull and rank(GG)} and   
Theorems \ref{th.self-orthogonal_equivence}-\ref{th.LCD_equivence}, the following corollaries are straightforward, 
which give several new necessary and sufficient conditions for a code being a Galois self-orthogonal code, a Galois self-dual code, or a Galois LCD code. 

\begin{corollary}\label{coro. Galois self-orthogonal code judgment}
    Let $\mathcal{C}$ be an $[n,k,d]_q$ code. Suppose that $G$ is any generator matrix of $\C$. 
    Then these four following statements are equivalent. 
    \begin{enumerate} 
        \item [\rm (1)] $\mathcal{C}$ is an $e$-Galois self-orthogonal code;
        \item [\rm (2)] $\mathcal{C}$ is an $(h-e)$-Galois self-orthogonal code;
        \item [\rm (3)] $G\sigma^e(G^T)$ is a zero matrix;
        \item [\rm (4)] $G\sigma^{h-e}(G^T)$ is a zero matrix.
    \end{enumerate}
\end{corollary}

\begin{corollary}\label{coro. Galois LCD code judgment}
    Let $\mathcal{C}$ be an $[n,k,d]_q$ code. Suppose that $G$ is any generator matrix of $\C$. 
    Then these four following statements are equivalent. 
    \begin{enumerate} 
        \item [\rm (1)] $\mathcal{C}$ is an $e$-Galois LCD code;
        \item [\rm (2)] $\mathcal{C}$ is an $(h-e)$-Galois LCD code;
        \item [\rm (3)] $G\sigma^e(G^T)$ is a nonsingular matrix; 
        \item [\rm (4)] $G\sigma^{h-e}(G^T)$ is a nonsingular matrix.
    \end{enumerate}
\end{corollary}

In particular, for Galois self-dual codes, the situation becomes a little more complicated. 
\begin{corollary}\label{coro. Galois self-dual code judgment}
    Let $\mathcal{C}$ be an $[n,k,d]_q$ code. Suppose that $G$ and $H$ are any generator matrix and 
    parity check matrix of $\C$, respectively. Then these six following statements are equivalent. 
    \begin{enumerate} 
        \item [\rm (1)] $\mathcal{C}$ is an $e$-Galois self-dual code;
        \item [\rm (2)] $\mathcal{C}$ is an $(h-e)$-Galois self-dual code;
        \item [\rm (3)] Both $G\sigma^e(G^T)$  and $\sigma^{h-e}(H)\sigma^{2h-2e}(H^T)$ are zero matrices;
        \item [\rm (4)] Both $G\sigma^e(G^T)$  and $\sigma^{h-e}(H)H^T$ are zero matrices;
        \item [\rm (5)] Both $G\sigma^{h-e}(G^T)$  and $\sigma^{h-e}(H)\sigma^{2h-2e}(H^T)$ are zero matrices;
        \item [\rm (6)] Both $G\sigma^{h-e}(G^T)$  and $\sigma^{h-e}(H)H^T$ are zero matrices.
    \end{enumerate}
\end{corollary}
\begin{proof}
    We only need to state a few facts, then the desired result follows from Theorem \ref{th.self-dual_equivence} and Corollary 
    \ref{coro. Galois self-orthogonal code judgment}. 
    \begin{itemize}
        \item $\textbf{Fact 1:}$ $\C$ is $e$-Galois self-dual implying $\C=\C^{\bot_e}$, i.e., $\C\subseteq \C^{\bot_e}$ and 
        $\C^{\bot_e}\subseteq \C=(\C^{\bot_e})^{\bot_{h-e}}$. Hence, $\C$ is $e$-Galois self-dual if and only if $\C$ is $e$-Galois 
        self-orthogonal and $\C^{\bot_e}$ is $(h-e)$-Galois self-orthogonal.  
        \item $\textbf{Fact 2:}$ It is easy to see that $\sigma^{h-e}(H)$ is a parity check matrix of $\sigma^{h-e}(\C)$. 
        Hence, $\sigma^{h-e}(H)$ is a generator matrix of $\C^{\bot_e}$. In summary, $G$ and $\sigma^{h-e}(H)$ are generator 
        matrices of $\C$ and $\C^{\bot_e}$, respectively.
        \item $\textbf{Fact 3:}$ These two following equalities hold: 
        \begin{align*}
            \begin{split}
                & \sigma^{h-e}(H)\sigma^{h-e}((\sigma^{h-e}(H))^{T})=\sigma^{h-e}(H)\sigma^{2h-2e}(H^T), \\
                & \sigma^{h-e}(H)\sigma^{e}((\sigma^{h-e}(H))^{T})=\sigma^{h-e}(H)H^T.
            \end{split}
        \end{align*}
    \end{itemize}
\end{proof}

In the following, we give some concrete examples to illustrate the relevant conclusions in this section. 

\begin{example}\label{example.111}
    Let $p=2$, $h=3$, i.e., $q=8$ and $\omega$ be a primitive element of $\mathbb{F}_{8}$. Let 
    \begin{align*}
        \begin{footnotesize}
        G= \left( 
            \begin{array}{cccccccc}
                    1 & 0 & 0 & 0 & 0 & \omega^4 & \omega^3 & \omega^2  \\
                    0 & 1 & 0 & 0 & 0 & 1 & \omega^2 & \omega^2 \\
                    0 & 0 & 1 & 0 & 0 & \omega^6 & \omega^2 & 1  \\
                    0 & 0 & 0 & 1 & 0 & \omega^4 & \omega^2 & 1 \\
                    0 & 0 & 0 & 0 & 1 & \omega & \omega^2 & \omega^4 \\
            \end{array}
            \right) \end{footnotesize}
    \end{align*}
    be a generator matrix of a $[8,5,3]_8$ code $\C$. 
    \begin{enumerate}
        \item [\rm (1)]  Computed with the Magma software package \cite{BCP1997}, we have the following results: 
        \begin{itemize}
            \item $\Hull_0(\C)=\Hull_3(\C)$ is a $[8,0,8]_8$ code, i.e., $\C$ is a $[8,5,3]_8$ Euclidean LCD code;
            \item $\Hull_1(\C)$ is a $[8,1,8]_8$ code with a generator matrix $(1\ 1\ \omega\ \omega^2\ \omega^3\ \omega^6\ \omega^6\ \omega^5)$, i.e., the dimension of $1$-Galois hull of $\C$ is $1$. 
            \item $\Hull_2(\C)$ is a $[8,1,7]_8$ code with a generator matrix $(1\ 0\ \omega\ \omega^3\ \omega\ \omega\ \omega^2\ \omega)$, i.e., the dimension of $2$-Galois hull of $\C$ is $1$;
        \end{itemize}
        Therefore, one can see that $\dim(\Hull_0(\C))=\dim(\Hull_3(\C))=0$ and $\dim(\Hull_1(\C))=\dim(\Hull_2(\C))=1$, which verify Theorem \ref{th.symmetry of Galois hulls}. 
    
        \item [\rm (2)]  Computed with the Magma software package \cite{BCP1997}, we also have the following results: 
        \begin{itemize}
            \item $\Hull_0(\C)=\Hull_0(\C^{\bot_0})$ is a $[8,0,8]_8$ code;
            \item $\Hull_1(\C)$ is a $[8,1,8]_8$ code with a generator matrix $(1\ 1\ \omega\ \omega^2\ \omega^3\ \omega^6\ \omega^6\ \omega^5)$ and 
                  $\Hull_1(\C^{\bot_1})$ is a $[8,1,7]_8$ code with a generator matrix $(1\ 0\ \omega^2\ \omega^6\ \omega^2\ \omega^2\ \omega^4\ \omega^2)$; 
            \item $\Hull_2(\C)$ is a $[8,1,7]_8$ code with a generator matrix $(1\ 0\ \omega\ \omega^3\ \omega\ \omega\ \omega^2\ \omega)$ and 
                  $\Hull_2(\C^{\bot_2})$ is a $[8,1,8]_8$ code with a generator matrix $(1\ 1\ \omega^4\ \omega\ \omega^5\ \omega^3\ \omega^3\ \omega^6)$. 
        \end{itemize}
        Therefore, one can see that $\dim(\Hull_0(\C))=\dim(\Hull_0(\C^{\bot_0}))=0$, $\dim(\Hull_1(\C))=\dim(\Hull_1(\C^{\bot_1}))=1$, 
        and $\dim(\Hull_2(\C))=\dim(\Hull_2(\C^{\bot_2}))=1$, which verify Corollary \ref{coro.C and C dual Hull equals}.         
    \end{enumerate}

\end{example}

\begin{example}\label{example.222}
    Let $p=2$, $h=3$, i.e., $q=8$ and $\omega$ be a primitive element of $\mathbb{F}_{8}$. Let 
    \begin{align*}
        \begin{footnotesize}
        G= \left( 
            \begin{array}{cccccccccc}
                    1 & 0 & 0 & 0 & \omega^4 & \omega^4 & 1 & \omega^2 & \omega^5 & 1 \\
                    0 & 1 & 0 & 0 & 1 & 1 & \omega^6 & \omega & \omega^5 & 0 \\
                    0 & 0 & 1 & 0 & 0 & 1 & 1 & \omega^6 & \omega & \omega^5 \\
                    0 & 0 & 0 & 1 & \omega^5 & \omega^3 & \omega^3 & \omega^3 & \omega^5 & 1\\
            \end{array}
            \right) \end{footnotesize}
    \end{align*}
    be a generator matrix of a $[10,4,6]_8$ code $\C$. Computed with the Magma software package \cite{BCP1997}, we have the following 
    results: 
    \begin{itemize}

       \item $\C$ is both a $1$-Galois LCD code and a $2$-Galois LCD code; 
        
       \item Both 
       \begin{align*}
        \begin{footnotesize}
            G\sigma^1(G^T)=\left(\begin{array}{cccc}
                \omega^4 & \omega^3 &  \omega^2 & \omega^4 \\
                0 & \omega^4 & 0 & \omega \\
                1 & \omega &  \omega^4 & \omega^4 \\
                \omega & \omega^3 & 0 & \omega^2 \\
            \end{array}\right)  \end{footnotesize}  {\rm and\ } 
        \begin{footnotesize}
            G\sigma^2(G^T)=\left(\begin{array}{cccc}
                \omega^2 & 0 &  1 & \omega^4 \\
                \omega^5 & \omega^2 & \omega^4 & \omega^5 \\
                \omega & 0 &  \omega^2 & 0 \\
                \omega^2 & \omega^4 & \omega^2 & \omega \\
            \end{array}\right)
        \end{footnotesize} 
      \end{align*}
        are nonsingular matrices. 
    \end{itemize}
Therefore, this example verifies Corollary \ref{coro. Galois LCD code judgment}. 

\end{example}

\begin{example}\label{example.333}
    Let $p=3$, $h=4$, i.e., $q=3^4$ and $\omega$ be a primitive element of $\mathbb{F}_{3^4}$. From \cite[Example 7.2]{RefJ52}, 
    taking $\lambda=w^{60}$, i.e., $\theta=w^5$, $\lambda=\theta^{12}$, $n=12$ and the generator polynomial 
    $g(x)=(x-\omega^5)^2(x-\omega^{25})(x-\omega^{45})^2(x-\omega^{65})$, we can get a $[12,6,3]_{3^4}$ $1$-Galois self-dual 
    $w^{60}$-constacyclic code $\mathcal{C}$. Computed with the Magma software package \cite{BCP1997}, we have the following 
    results: 
    \begin{itemize}
        \item A generator matrix $G$ and a parity check matrix $H$ of $\C$ are given by  
        \begin{align*}
           \begin{footnotesize}
            G= \left( 
                \begin{array}{cccccccccccc}
                        1 & 0 & 0 & 0 & 0 & 0 & \omega^{50} & 0 & 1 &0 & \omega^{70} & 0 \\
                        0 & 1 & 0 & 0 & 0 & 0 & 0 & \omega^{50} & 0 & 1 & 0 & \omega^{70} \\
                        0 & 0 & 1 & 0 & 0 & 0 & \omega^{60} & 0 & 0 & 0 & 2 & 0 \\
                        0 & 0 & 0 & 1 & 0 & 0 & 0 & \omega^{60} & 0 & 0 & 0 & 2 \\
                        0 & 0 & 0 & 0 & 1 & 0 & \omega^{30} & 0 & \omega^{20} & 0 & \omega^{50} & 0 \\
                        0 & 0 & 0 & 0 & 0 & 1 & 0 & \omega^{30} & 0 & \omega^{20} & 0 & \omega^{50} \\
                \end{array}
                \right) \end{footnotesize} {\rm and}  
        \end{align*}
                    
        \begin{align*}
            \begin{footnotesize}
             H= \left( 
                \begin{array}{cccccccccccc}
                            1 & 0 & 0 & 0 & 0 & 0 & \omega^{70} & 0 & 1 &0 & \omega^{50} & 0 \\
                            0 & 1 & 0 & 0 & 0 & 0 & 0 & \omega^{70} & 0 & 1 & 0 & \omega^{50} \\
                            0 & 0 & 1 & 0 & 0 & 0 & \omega^{20} & 0 & 0 & 0 & 2 & 0 \\
                            0 & 0 & 0 & 1 & 0 & 0 & 0 & \omega^{20} & 0 & 0 & 0 & 2 \\
                            0 & 0 & 0 & 0 & 1 & 0 & \omega^{10} & 0 & \omega^{60} & 0 & \omega^{70} & 0 \\
                            0 & 0 & 0 & 0 & 0 & 1 & 0 & \omega^{10} & 0 & \omega^{60} & 0 & \omega^{70} \\
                \end{array} 
                \right) \end{footnotesize},
        \end{align*}
        respectively;
        
        \item $\C$ is both a $1$-Galois self-dual code and a $3$-Galois self-dual code; 
        
        \item $G\sigma^1(G^T)=G\sigma^3(G^T)=\sigma^3(H)\sigma^6(H^T)=\sigma^3(H)H^T=O_{6\times 6}$.
        \end{itemize}
    Therefore, this example verifies Corollary \ref{coro. Galois self-dual code judgment}. 
        
\end{example}

\section{Codes with Galois hulls of arbitrary dimensions from Galois self-orthogonal codes}\label{sec4}
In this section, we want to construct codes of larger length with Galois hulls of arbitrary dimensions from 
a given Galois self-orthogonal code. According to mathematical induction, it seems that we only need to consider the case 
where the length increases by $1$, then larger lengths can be obtained similarly. However, in Remark \ref{rem8.length of n+1} 
of Section \ref{sec4.4}, we can see that codes of length $n+1$ with $k$-dimensional Galois hull can not be deduced from a given 
$[n,k,d]_q$ Galois self-orthogonal code by the method introduced in Theorem \ref{th.arbitrary hull codess with length n+1}.  
Hence, for our topic, we first prove that Galois self-orthogonal codes of length $n+i$ ($i\geq 0$ and $i\neq 1$) can be derived from Galois self-orthogonal codes of length $n$. 
Then, based on these Galois self-orthogonal codes, we derive codes with Galois hulls of arbitrary dimensions. 
Finally, many interesting examples are given. 

\subsection{Galois self-orthogonal codes of length $n+i$ $(i\geq 0$ and $i\neq 1)$ from Galois self-orthogonal codes of length $n$}\label{sec4.1}

\begin{theorem}\label{th.n+2i}
    Let $\mathcal{C}$ be an $[n,k,d]_q$ $e$-Galois self-orthogonal code. For any integer $i\geq 0$ and $i\neq 1$, if there exist 
    $\alpha_1,\alpha_2,\cdots,\alpha_{i}\in \mathbb{F}_q^*$ such that 
    \begin{align}\label{eq.th.n+i}
        \alpha_1^{p^{e}+1}+\alpha_2^{p^{e}+1}+\dots+\alpha_{i}^{p^{e}+1}=0, 
    \end{align}
    then there exists an $[n+i,k,\widetilde{d}_{i}]_q$ $e$-Galois self-orthogonal code $\mathcal{C}_{i}$, where $d\leq \widetilde{d}_{i}\leq n+i+1-k$.
\end{theorem}
\begin{proof}
    Let $G$ be a generator matrix of $\mathcal{C}$, then up to equivalence, we can set $$G= \left(\begin{array}{c|c}
        I_k & A 
    \end{array} \right).$$ 
    From Corollary \ref{coro. Galois self-orthogonal code judgment}, since $\mathcal{C}$ is an $e$-Galois self-orthogonal code, we have 
    $G\sigma^e(G^T)=O_{k\times k}$, i.e., $A\sigma^{e}(A^T)=-I_k$. Choose $\beta\in \mathbb{F}_q^*$ such that $\beta^{p^e+1}=1$ and note that  
    $\beta=1$ is always an appropriate choice. For convenience, denote $\lambda=\alpha_1^{p^{e}+1}+\alpha_2^{p^{e}+1}+\dots+\alpha_{i}^{p^{e}+1}$, 
    then $\lambda=0$. 
    
    We now construct $e$-Galois self-orthogonal codes of length $n+i$ $(i\geq 2)$. Let   
    \begin{equation*}
        G_{i}= \left(                 
             \begin{array}{c|c|c|c|c|c}   
                G' & A & \mathbf{a_1} & \mathbf{a_2} & \cdots & \mathbf{a_{i}}
             \end{array} 
         \right)               
    \end{equation*}
    be a generator matrix of a code $\C_{i}$, 
    where $G'=\beta I_k={\rm diag}(\underbrace{\beta,\beta,\cdots,\beta}_k)$ and 
    $\mathbf{a_s}=(\underbrace{\alpha_s\ \alpha_s\ \cdots\ \alpha_s}_k)^T$ 
    for $1\leq s\leq i$. Then $\dim(\C_{i})=\dim(\C)=k$ and 
    \begin{align*}
        G_{i}\sigma^{e}(G_{i}^T) & = \left(                 
        \begin{array}{c|c|c|c|c|c}   
           G' & A & \mathbf{a_1} & \mathbf{a_2} & \cdots & \mathbf{a_{i}}
        \end{array} 
    \right)
    \left(                 
        \begin{array}{c}   
           \sigma^e(G'^T) \\ \hline 
           \sigma^e(A^T) \\ \hline
           \sigma^e(\mathbf{a_1}^T) \\ \hline
           \sigma^e(\mathbf{a_2}^T) \\ \hline
           \vdots \\ \hline
           \sigma^e(\mathbf{a_{i}}^T) \\
        \end{array}
    \right)\\
   &  =  G'\sigma^e(G'^T)+A\sigma^e(A^T)+\mathbf{a_1}\sigma^e(\mathbf{a_1}^T)+\mathbf{a_2}\sigma^e(\mathbf{a_2}^T)+
   \cdots+\mathbf{a_{i}}\sigma^e(\mathbf{a_{i}}^T) \\
   &  =  {\rm diag}(\underbrace{\beta^{p^{e}+1}-1,\beta^{p^{e}+1}-1,\cdots,\beta^{p^{e}+1}-1}_k)+\lambda \cdot \mathbf{1}_{k\times k} \\
   &  =  O_{k\times k}.
    \end{align*}
    From Corollary \ref{coro. Galois self-orthogonal code judgment} again, we can see that $\C_{i}$ ($i\geq 2$) is an $e$-Galois self-orthogonal code of length $n+i$ $(i\geq 2)$. 

    Next, 
    we show that $d\leq \widetilde{d}_{i}\leq n+i+1-k$ for $i\geq 2$. 
    Let $\mathbf{c} \in \C_{i}\setminus \{\mathbf{0}_{n+i}\}$. 
    Then there exists a $\mathbf{u}\in \F_q^k\setminus \{\mathbf{0}_{n+i}\}$ such that 
    $$\mathbf{c}=\mathbf{u}G_{i}=(\mathbf{u}G'\ |\ \mathbf{u}A\ |\ \mathbf{u}\mathbf{a_1}\ |\ \mathbf{u}\mathbf{a_2}\ |\ \cdots\ |\ \mathbf{u}\mathbf{a_{i}}).$$
    Note that $\omega_H((\mathbf{u}G'\ |\ \mathbf{u}A))=\omega_H((\mathbf{u}\beta I_k\ |\ \mathbf{u}A))=\omega_H(\mathbf{u}(I_k\ |\ A))=\omega_H(\mathbf{u}G)\geq d$. 
    Hence,  $\omega_H(\mathbf{c})$ can be estimated as follows: 
    \begin{align*}
        \begin{split}
            \omega_H(\mathbf{c}) & = \omega_H((\mathbf{u}G'\ |\ \mathbf{u}A\ |\ \mathbf{u}\mathbf{a_1}\ |\ \mathbf{u}\mathbf{a_2}\ |\ \cdots\ |\ \mathbf{u}\mathbf{a_{i}})) \\
                                 & \geq \omega_H((\mathbf{u}G'\ |\ \mathbf{u}A)) \\
                                 & \geq d,
        \end{split}
    \end{align*}
    which implies that $\widetilde{d}_{i}\geq d$. And according to the Singleton bound, we have $\widetilde{d}_{i}\leq n+i-k+1$. Hence, $d\leq \widetilde{d}_{i}\leq n+i+1-k$, where $i\geq 2$.   

    Finally, we consider the case $i=0$. It is easy to check that, in this case, Equation (\ref{eq.th.n+i}) is no longer a restriction and 
    $G_0=(G'\ |\ A)$ can generate an $e$-Galois self-orthogonal code $\C_0$, which is equivalent to $\C$. 
    It is well known that equivalent codes have the same parameters. Therefore, we complete the proof. 
\end{proof}

\begin{remark}\label{rem5.th.n+2i} $\quad$
    \begin{enumerate}
        \item [\rm (1)] By adding columns with all zero elements to a generator matrix, one can easily obtain $e$-Galois self-orthogonal codes 
        of larger length from a given $e$-Galois self-orthogonal code. However, according to the proof of Theorem \ref{th.n+2i}, it is easy 
        to see that this way keeps the minimum distance of new codes unchanged. And hence, the method in Theorem \ref{th.n+2i} can result in 
        codes with better parameters. For this reason, in the following related theorems, we also require that the elements in the added 
        column can not be all zero.

        \item [\rm (2)] Note that $i\neq 1$ is needed in Theorem \ref{th.n+2i}. In fact, if we take $i=1$, then Equation (\ref{eq.th.n+i}) yields $\alpha^{p^e+1}=0$, 
        which contradicts to the condition $\alpha_1\in \F_q^*$ in Theorem \ref{th.n+2i}. 
        Hence, $i\neq 1$ is necessary for Theorem \ref{th.n+2i}.

    \end{enumerate}
\end{remark}

\subsection{Conditions are relatively weak}\label{sec4.2}

Formally, the condition given in Equation (\ref{eq.th.n+i}) is complicated. In this subsection, we show that 
Equation (\ref{eq.th.n+i}) holds under certain cases. Note that extensive searches based on the Magma software 
package \cite{BCP1997} illustrate that examples satisfying these cases are comparatively common. 
Therefore, we can conclude that the conditions in Theorem \ref{th.n+2i} are actually relatively weak.

\begin{lemma}\label{lemma.gcd}{\rm (\cite{RefJ51})}
    Let $s\geq 1$ and $p>1$ be two integers. Then
    \begin{equation}\label{equ.gcd}
    \gcd(p^r+1, p^s-1)=\left\{
    \begin{array}{rl}
    1, & {\rm {if}}\ \frac{s}{\gcd(r,s)}\ {\rm {is\ odd\ and}}\ p\ {\rm is\ even,}\\
    2, & {\rm {if}}\ \frac{s}{\gcd(r,s)}\ {\rm {is\ odd\ and}}\ p\ {\rm is\ odd,}\\
    p^{\gcd(r,s)}+1, & {\rm {if}}\ \frac{s}{\gcd(r,s)}\ {\rm {is\ even.}}\\
    \end{array} \right.
    \end{equation}
\end{lemma}

\begin{lemma}\label{lem.t_0 exists}
    Let $\mathbb{F}_q^*=\langle \omega \rangle$, where $q=p^h$ is an odd prime power and $\omega$ is a primitive element of $\mathbb{F}_q$. 
    Then there is an integer $t_0$ such that $\omega^{t_0(p^{e}+1)}=-1$ if either of the following conditions holds. 
    \begin{enumerate}
        \item [\rm (1)] $\frac{h}{\gcd(e,h)}$ is odd and  $\frac{p^h-1}{2}$ is even;
        \item [\rm (2)] $\frac{h}{\gcd(e,h)}$ is even and $(p^{\gcd(e,h)}+1)\mid \frac{p^h-1}{2}$.
    \end{enumerate}
\end{lemma}
    \begin{proof}
        Since $q=p^h$ is an odd prime power, we have $\omega^{\frac{p^h-1}{2}}=-1$. Hence, for our target, it suffices to prove that 
        there is an integer $t_0$ such that 
        \begin{equation}\label{eq.t_0 exists}
            t_0(p^{e}+1)\equiv \frac{p^h-1}{2}({\rm mod}\ p^h-1).        
        \end{equation}
    
        Denote $d=\gcd(p^{e}+1,p^h-1)$ and see Equation (\ref{eq.t_0 exists}) as a congruence equation of $t_0$. Then, according to basic number 
        theory knowledge, $t_0$ exists as a solution if and only if $d\mid \frac{p^h-1}{2}$. We consider the two given conditions as follows: 
        \begin{itemize}
            \item $\textbf{Condition (1)}$: When $\frac{h}{\gcd(e,h)}$ is odd, it follows from odd $p$ and Lemma \ref{lemma.gcd} that $d=2$. Hence, $t_0$ exists 
            as a solution if and only if $\frac{p^h-1}{2}$ is even.
            \item $\textbf{Condition (2)}$: When $\frac{h}{\gcd(e,h)}$ is even, it follows from Lemma \ref{lemma.gcd} that $d=p^{\gcd(e,h)}+1$. Hence, $t_0$ exists 
            as a solution if and only if $(p^{\gcd(e,h)}+1)\mid \frac{p^h-1}{2}$.
        \end{itemize}

        Therefore, the desired result follows.
    \end{proof}

    Next, we characterize the conditions under which Equation (\ref{eq.th.n+i}) holds. 
    To be more explicit, we divide $n+i$ ($i\geq 0$ and $i\neq 1$) into two parts $n+2i$ ($i\geq 0$) and $n+2i+1$ ($i\geq 1$).
    
    \begin{theorem}\label{th.alpha_1,alpha_2,...,alpha_{2i} exist}
        If one of the following conditions is met: 
        \begin{enumerate}
            \item [\rm (1)] $p$ is even;
            \item [\rm (2)] $p$ is odd, $\frac{h}{\gcd(e,h)}$ is odd and  $\frac{p^h-1}{2}$ is even;
            \item [\rm (3)] $p$ is odd, $\frac{h}{\gcd(e,h)}$ is even and $(p^{\gcd(e,h)}+1)\mid \frac{p^h-1}{2}$,
        \end{enumerate}
        then the following statements hold. 
        \begin{enumerate}
            \item [\rm (1)] For $i\geq 0$, there exist $\alpha_1,\alpha_2,\cdots,\alpha_{2i} \in \mathbb{F}_q^*$ such that 
            \begin{align}\label{eq.n+2i.equals 0}
               \alpha_1^{p^{e}+1}+\alpha_2^{p^{e}+1}+\cdots+\alpha_{2i}^{p^{e}+1}=0.
            \end{align} 

            \item [\rm (2)] For $i\geq 1$, if there further exist $\alpha_1,\alpha_2,\alpha_3\in \mathbb{F}_q^*$ satisfying 
            $\alpha_1^{p^{e}+1}+\alpha_2^{p^{e}+1}+\alpha_3^{p^{e}+1}=0$, then there exist $\alpha_4,\alpha_5,\cdots,\alpha_{2i+1}\in \mathbb{F}_q^*$ such that 
            \begin{align}\label{eq.n+2i+1.equals 0}
               \alpha_1^{p^{e}+1}+\alpha_2^{p^{e}+1}+\cdots+\alpha_{2i+1}^{p^{e}+1}=0.
            \end{align} 
        \end{enumerate}
    \end{theorem}
    \begin{proof}
        (1) Clearly, $i=0$ is a trivial case. For $i\geq 1$, we note that $\alpha_1,\alpha_2,\cdots,\alpha_{2i}$ can be the same. 
        Hence, if we can determine that there exist $\alpha_1,\alpha_2\in \mathbb{F}_q^*$ such that $\alpha_1^{p^{e}+1}+\alpha_2^{p^{e}+1}=0$, 
        then Equation (\ref{eq.n+2i.equals 0}) holds for any $i\geq 1$. 
        Now, we directly give some possible values of $\alpha_1$ and $\alpha_2$ under different conditions. 

        \begin{itemize}
            \item $\textbf{Condition (1)}$: It follows from even $p$ that $2\alpha^{p^{e}+1}=0$ for any $\alpha \in \mathbb{F}_q$. 
            Hence, we can take $\alpha_1=\alpha_2\in \mathbb{F}_q^*$, then $\alpha_1^{p^{e}+1}+\alpha_2^{p^{e}+1}=2\alpha_1^{p^{e}+1}=0$. 

            \item $\textbf{Condition (2) or (3)}$: From Lemma \ref{lem.t_0 exists}, there is an integer $t_0$ such that 
            $\omega^{t_0(p^{e}+1)}=-1$. Hence, we can take $\alpha_1=\omega^{t_0}\in \mathbb{F}_q^*$ and some 
            $\alpha_2\in \mathbb{F}_q^*$ satisfying $\alpha_{2}^{p^e+1}=1$ 
            ($\alpha_2=1$ is trivial), then $\alpha_1^{p^{e}+1}+\alpha_2^{p^{e}+1}=-1+1=0$. 
        \end{itemize}

        Therefore, we complete the proof of the result (1). 

        (2) Clearly, $i=1$ is a trivial case. For $i\geq 2$, the desired result follows from the fact $2i+1=3+2(i-1)$ 
        and the proof given in (1) above.  
    \end{proof}

    \begin{remark}\label{rem7.th18.a1,...,a_{2i} exist} $\quad$
        \begin{enumerate}
            \item [\rm (1)] Obviously, there also exist other possible $\alpha_s \in \mathbb{F}_q^*$ for $1\leq s\leq 2i$ satisfying 
            $\alpha_r\neq \alpha_t$ with any $1\leq r\neq t\leq 2i$ such that Equation (\ref{eq.n+2i.equals 0}) holds in some cases. 
            For example, taking $q=2^3$, $i=2$, searching by the Magma software package \cite{BCP1997}, 
            \begin{itemize}
                \item when $e=0$, we can take $\alpha_1=1,\ \alpha_2=\omega,\ \alpha_3=\omega^6$ and $\alpha_4=\omega^4$; 
                \item when $e=1$, we can take $\alpha_1=1,\ \alpha_2=\omega^5,\ \alpha_3=\omega^2$ and $\alpha_4=\omega^6$; 
                \item when $e=2$, we can take $\alpha_1=\omega^6,\ \alpha_2=1,\ \alpha_3=\omega$ and $\alpha_4=\omega^3$. 
            \end{itemize}

            \item [\rm (2)] Conditions (1)-(3) in Theorem \ref{th.alpha_1,alpha_2,...,alpha_{2i} exist} are only sufficient. Therefore, 
            for the case where all these conditions are not satisfied, we can not be completely sure that 
            there exist no $\alpha_s\in \mathbb{F}_q^*$ for $1\leq s\leq 2i+1$ satisfying Equations (\ref{eq.n+2i.equals 0}) and (\ref{eq.n+2i+1.equals 0}). 
            For example, take $q=3^3$, then for any $0\leq e\leq 2$, both $\frac{h}{\gcd(e,h)}=1$ or $3$ and $\frac{p^h-1}{2}=13$ are odd. 
            Searching by the Magma software package \cite{BCP1997}, we can find that 
            \begin{itemize}
                \item for $i=1$, Equation (\ref{eq.n+2i.equals 0}) does not hold; 
                \item for $i=2$, Equation (\ref{eq.n+2i.equals 0}) can indeed be established. 
                For example,
                \begin{itemize}
                    \item when $e=0$, we can take $\alpha_1=w^3,\ \alpha_2=\omega^{17},\ \alpha_3=1$ and $\alpha_4=\omega^{11}$; 
                    \item when $e=1$, we can take $\alpha_1=w^4,\ \alpha_2=\omega^{5},\ \alpha_3=1$ and $\alpha_4=\omega^{7}$; 
                    \item when $e=2$, we can take $\alpha_1=w^6,\ \alpha_2=\omega^{11},\ \alpha_3=1$ and $\alpha_4=\omega^{10}$. 
                \end{itemize}  
                Moreover, we can conclude that Equation (\ref{eq.n+2i.equals 0}) must hold for each $i$ satisfying $2\mid i$. 
            \end{itemize}

            \item [\rm (3)] Again, the condition ``$\alpha_1^{p^{e}+1}+\alpha_2^{p^{e}+1}+\alpha_3^{p^{e}+1}=0$ holds for some $\alpha_1,\alpha_2,\alpha_3\in \mathbb{F}_q^*$'' is also only sufficient. 
            For example, take $q=5$, $e=0$, then there exist no $\alpha_s\in \mathbb{F}_q^*$ for $1\leq s\leq 3$ such that $\alpha_1^2+\alpha_2^2+\alpha_3^2=0$. 
            However, taking $\alpha_s=1$ for $1\leq s\leq 5$, we have $\alpha_1^2+\alpha_2^2+\cdots+\alpha_5^2=0$.    

        \end{enumerate}
    \end{remark}

\begin{example}\label{example.e and h values}
    Denote a primitive element of $\mathbb{F}_{p^h}$ by $\omega_{p,h}$. 
    In order to illustrate the effect of Theorem \ref{th.alpha_1,alpha_2,...,alpha_{2i} exist} more intuitively, 
    we list some examples satisfying the Condition (2) or (3) in Table \ref{tab:1} and some examples satisfying $\alpha_1^{p^e+1}+\alpha_2^{p^e+1}+\alpha_3^{p^e+1}=0$ in Table \ref{tab:2}. 
    
    \begin{enumerate}
        \item [\rm (1)] Table \ref{tab:1} shows that when we take $p=5$ and $p=13$ with $1\leq h\leq 4$, examples satisfying the 
        Condition $(2)$ or $(3)$ in Theorem \ref{th.alpha_1,alpha_2,...,alpha_{2i} exist} can take all possible $h$ and $e$. 
        Note that when $p=3$, we can see that $h\neq 1,3$. However, one can still conclude that there are many more cases where the Condition (2) or (3) is met than where they are not.
    
        \item [\rm (2)] Table \ref{tab:2} shows that when we take $p=3$ with $1\leq h\leq 4$, 
        $\alpha_1^{p^e+1}+\alpha_2^{p^e+1}+\alpha_3^{p^e+1}=0$ with some 
        $\alpha_1,\alpha_2,\alpha_3\in \mathbb{F}_{3^h}^*$ holds for all possible $h$ and $e$. 
    \end{enumerate}
\end{example}

\begin{table}[!htb]
           \caption{Some examples satisfying the Condition $(2)$ or $(3)$ in Theorem \ref{th.alpha_1,alpha_2,...,alpha_{2i} exist} for $p=3$, $5$ and $13$ with $1\leq h\leq 4$}
       \label{tab:1}       
       \begin{center}
           \begin{tabular}{c|c|l||c|c|l}
            \hline
              $p$ & $h$ & $e$ & $p$ & $h$ & $e$\\ \hline 
               3 & 2 & 0,1 & 3 & 4 & 0,1,2,3 \\
               
               5 & 1 & 0 & 5 & 2 & 0,1 \\
               5 & 3 & 0,1,2 & 5 & 4 & 0,1,2,3 \\
               13 & 1 & 0 & 13 & 2 & 0,1 \\
               13 & 3 & 0,1,2 & 13 & 4 & 0,1,2,3 \\ \hline
           \end{tabular}
       \end{center}
\end{table}

\begin{table}[H]
    \caption{Some examples satisfying $\alpha_1^{p^e+1}+\alpha_2^{p^e+1}+\alpha_3^{p^e+1}=0$ for $p=3$ with $1\leq h\leq 4$}
    \label{tab:2}       
    \begin{center}
        \begin{tabular}{cccccc||cccccc}
            \hline
           $p$ & $h$ & $e$ & $\alpha_1$ & $\alpha_2$ & $\alpha_3$ & $p$ & $h$ & $e$ & $\alpha_1$ & $\alpha_2$ & $\alpha_3$\\\hline 
            3 & 1 & 0 & 1 & 1 & 1 & 3 & 2 & 0 & 2 & 1 & 1 \\
            3 & 2 & 1 & 1 & 2 & $\omega_{3,2}^6$ & 3 & 3 & 0 & 1 & 2 & 1 \\
            3 & 3 & 1 & $\omega_{3,3}^2$ & $\omega_{3,3}^7$ & 1 & 3 & 3 & 2 & $\omega_{3,3}^5$ & $\omega_{3,3}^{11}$ & 1 \\
            3 & 4 & 0 & $\omega_{3,4}^{6}$ & $\omega_{3,4}^{54}$ & $\omega_{3,4}^{5}$ & 3 & 4 & 1 & $\omega_{3,4}^{42}$ & $\omega_{3,4}^{22}$ & $\omega_{3,4}^{2}$ \\
            3 & 4 & 2 & $\omega_{3,4}$ & $\omega_{3,4}^{6}$ & 1 & 3 & 4 & 3 & $\omega_{3,4}^{21}$ & $\omega_{3,4}^{41}$ & $\omega_{3,4}$ \\
            \hline
        \end{tabular}
    \end{center}
\end{table}

\subsection{Further decisions for the Euclidean and Hermitian inner products}\label{sec4.3}

As shown in Theorem \ref{th.alpha_1,alpha_2,...,alpha_{2i} exist}, many cases can lead to Equations (\ref{eq.n+2i.equals 0}) and (\ref{eq.n+2i+1.equals 0}) being valid. 
Therefore, we can affirm that conditions in Theorem \ref{th.alpha_1,alpha_2,...,alpha_{2i} exist}  are relatively weak, i.e., the conditions in Theorem \ref{th.n+2i} are 
relatively weak. However, according to Remarks \ref{rem7.th18.a1,...,a_{2i} exist} (2) and (3), we can not give a full decision for all cases of $0\leq e\leq h-1$. 
In the following, we focus on the Euclidean and Hermitian inner products, and further refine these conditions. 

\begin{lemma}\label{Coro.Euclidean}
    Let $\mathcal{C}$ be an $[n,k,d]_q$ Euclidean self-orthogonal code. Then the following statements hold. 
    \begin{enumerate}
        \item [\rm (1)] If $p$ is even, or $p$ is odd and $h$ is even, there exists an $[n+2i,k,\widetilde{d}_{2i}]_q$ 
        Euclidean self-orthogonal code $\mathcal{C}_{2i}$ for each $i\geq 0$, where $d\leq \widetilde{d}_{2i}\leq n+2i+1-k$.  
        \item [\rm (2)] If $p$ is even and $h\geq 2$, or $p$ is odd and $h$ is even, 
        there exists an $[n+2i+1,k,\widetilde{d}_{2i+1}]_q$ Euclidean self-orthogonal code $\mathcal{C}_{2i+1}$ for each $i\geq 1$, 
        where $d\leq \widetilde{d}_{2i+1}\leq n+2(i+1)-k$. 
    \end{enumerate}
\end{lemma}
\begin{proof}
Note that the Euclidean self-orthogonal property implies that $e=0$. 

(1) For the proof of the result $(1)$, we can discuss it in terms of the parity of $p$. 
\begin{itemize}
    \item $\textbf{Case 1:}$ When $p$ is even, according to Theorem \ref{th.alpha_1,alpha_2,...,alpha_{2i} exist}, it is trivial. 
    \item $\textbf{Case 2:}$ When $p$ is odd and $h$ is even, it is easy to check that $\frac{h}{\gcd(0,h)}=1$ is odd and $\frac{p^h-1}{2}$ 
    is even. Hence, the Condition (2) in Theorem \ref{th.alpha_1,alpha_2,...,alpha_{2i} exist} holds. 
\end{itemize}

From Theorem \ref{th.alpha_1,alpha_2,...,alpha_{2i} exist} (1), we know that there exists an $[n+2i,k,\widetilde{d}_{2i}]_q$ 
Euclidean self-orthogonal code $\mathcal{C}_{2i}$ for each $i\geq 0$, where $d\leq \widetilde{d}_{2i}\leq n+2i+1-k$, which completes the proof. 

(2) Combining the proof of the result $(1)$ above, it suffices to prove that there always exist $\alpha_1,\alpha_2,\alpha_3\in \mathbb{F}_{p^h}^*$ 
such that $\alpha_1^2+\alpha_2^2+\alpha_3^2=0$ under the given conditions. We complete the proof in two ways. 
\begin{itemize}
    \item $\textbf{Case 1:}$ When $p$ is even and $h\geq 2$, we can choose $\alpha_1,\alpha_2\in \mathbb{F}_{2^h}^*$ 
    satisfying $\alpha_1+\alpha_2\in \F_{2^h}^*$, i.e., $\alpha_1^2+\alpha_2^2=(\alpha_1+\alpha_2)^2\neq 0$. 
    Set $\alpha_3=\alpha_1+\alpha_2\in \F_{2^h}^*$
    It follows that $\alpha_1^2+\alpha_2^2+\alpha_3^2=(\alpha_1+\alpha_2+\alpha_3)^2=0$ from the fact $\alpha_1+\alpha_2+\alpha_3=2(\alpha_1+\alpha_2)=0$.

    \item $\textbf{Case 2:}$ When $p$ is odd and $h$ is even, we further divide the discussion into two cases as follows. 
    \begin{itemize}
        \item $\textbf{Case 2.1:}$ When $p=3$ and $h=2$, we can take $\alpha_1=\alpha_2=\alpha_3=1\in \F_9^*$, then 
        $\alpha_1^2+\alpha_2^2+\alpha_3^2=3=0$.

        \item $\textbf{Case 2.1:}$ Otherwise, we have $(p^{\frac{h}{2}}-1)\nmid 2$. Then, we can choose 
        $\alpha_1 \in \mathbb{F}_{p^{\frac{h}{2}}}^* \subseteq\mathbb{F}_{p^h}^*$ satisfying $\alpha_1^{2}\neq 1$, which 
        follows that $1-\alpha_1^2\in \mathbb{F}_{p^\frac{h}{2}}^*$ is a square element in $\mathbb{F}_{p^h}$, and thus, there is an 
        $\alpha_2 \in \mathbb{F}_{p^h}^*$ such that $\alpha_2^2=1-\alpha_1^2$, i.e., $\alpha_1^2+\alpha_2^2=1$. Note that the Condition (2) 
        in Theorem \ref{th.alpha_1,alpha_2,...,alpha_{2i} exist} holds, then by Lemma \ref{lem.t_0 exists}, 
        there is an integer $t_0$ such that $\omega^{2t_0}=-1$, where $\omega$ is a primitive element of $\mathbb{F}_{p^h}$. 
        Denote $\alpha_3=\omega^{t_0}\in \mathbb{F}_{p^h}^*$, then $\alpha_1^2+\alpha_2^2+\alpha_3^2=1-1=0$.
    
    \end{itemize}
    \end{itemize}

From Theorem \ref{th.alpha_1,alpha_2,...,alpha_{2i} exist} (2), we know that there exists an $[n+2i+1,k,\widetilde{d}_{2i+1}]_q$ 
Euclidean self-orthogonal code $\mathcal{C}_{2i+1}$ for each $i\geq 1$, where $d\leq \widetilde{d}_{2i+1}\leq n+2(i+1)-k$, which completes the proof. 
\end{proof}

\begin{lemma}\label{Coro.Hermitian}
    Let $\mathcal{C}$ be an $[n,k,d]_q$ Hermitian self-orthogonal code. Then the following statements hold. 

    \begin{enumerate}
        \item [\rm (1)] There exists an $[n+2i,k,\widetilde{d}_{2i}]_q$ Hermitian self-orthogonal code $\mathcal{C}_{2i}$ 
        for each $i\geq 0$, where $d\leq \widetilde{d}_{2i}\leq n+2i+1-k$.  

        \item [\rm (2)] If $q>4$, there exists an $[n+2i+1,k,\widetilde{d}_{2i+1}]_q$ Hermitian self-orthogonal code $\mathcal{C}_{2i+1}$ 
        for each $i\geq 1$, where $d\leq \widetilde{d}_{2i+1}\leq n+2(i+1)-k$. 
    \end{enumerate}
\end{lemma}

\begin{proof}
Note that the Hermitian self-orthogonal property implies that $h$ is even and $e=\frac{h}{2}\neq 0$. 

(1) Taking a similar argument to Lemma \ref{Coro.Euclidean} (1), we complete the proof in the following two ways.  
\begin{itemize}
    \item $\textbf{Case 1:}$  When $p$ is even, according to Theorem \ref{th.alpha_1,alpha_2,...,alpha_{2i} exist}, it is trivial. 
    \item $\textbf{Case 2:}$ When $p$ is odd, it is easy to check that $\frac{h}{\gcd(\frac{h}{2},h)}=2$ is even and $p^{\gcd(\frac{h}{2},h)}+1=p^{\frac{h}{2}}+1$.  
    Note that $\frac{p^{\frac{h}{2}}-1}{2}$ is an integer. And since $p^h-1=(p^\frac{h}{2}-1)(p^\frac{h}{2}+1)$, 
    we have $\frac{p^h-1}{2}=\frac{p^\frac{h}{2}-1}{2}\cdot (p^\frac{h}{2}+1)$, i.e., $(p^{\frac{h}{2}}+1)\mid \frac{p^h-1}{2}$. Hence, the Condition $(3)$ in Theorem
    \ref{th.alpha_1,alpha_2,...,alpha_{2i} exist} holds. 
\end{itemize}

From Theorem \ref{th.alpha_1,alpha_2,...,alpha_{2i} exist} (1), we know that there exists an $[n+2i,k,\widetilde{d}_{2i}]_q$ 
Hermitian self-orthogonal code $\mathcal{C}_{2i}$ for each $i\geq 0$, where $d\leq \widetilde{d}_{2i}\leq n+2i+1-k$, which completes the proof. 

(2) Similar to Lemma \ref{Coro.Euclidean} (2), in combination with (1) above, it suffices to prove that there always exist 
$\alpha_1,\alpha_2,\alpha_3\in \mathbb{F}_{p^h}^*$ such that 
$
    \alpha_1^{p^{\frac{h}{2}}+1}+\alpha_2^{p^{\frac{h}{2}}+1}+\alpha_3^{p^{\frac{h}{2}}+1}=0
$ 
under the given conditions. We complete the proof in two ways. 

\begin{itemize}
    \item $\textbf{Case 1:}$ When $p$ is even and $q=2^h>4$, we can choose $\alpha_1,\alpha_2\in \mathbb{F}_{2^h}^*$ satisfying 
    $
        \alpha_1^{2^\frac{h}{2}+1}+\alpha_2^{2^\frac{h}{2}+1}\neq 0
    $. 
    Otherwise, taking $\alpha_1^{2^\frac{h}{2}+1}=1$, then $\alpha_2^{2^\frac{h}{2}+1}=1$ for any $\alpha_2\in \mathbb{F}_{2^h}^*$, which contradicts to 
    $(2^h-1)\nmid (2^{\frac{h}{2}}+1)$. Furthermore, it follows from  
    \begin{align*}
        (\alpha_1^{2^\frac{h}{2}+1}+\alpha_2^{2^\frac{h}{2}+1})^{2^\frac{h}{2}}
        =\alpha_1^{2^\frac{h}{2}+1}+\alpha_2^{2^\frac{h}{2}+1}  
    \end{align*}
    
    that $\alpha_1^{2^\frac{h}{2}+1}+\alpha_2^{2^\frac{h}{2}+1}\in \mathbb{F}_{2^\frac{h}{2}}^*$. Hence, there is an $\alpha_3\in \mathbb{F}_{2^h}^*$ 
    such that $\alpha_3^{2^\frac{h}{2}+1}=\alpha_1^{2^\frac{h}{2}+1}+\alpha_2^{2^\frac{h}{2}+1}$, which implies that  
    $\alpha_1^{2^\frac{h}{2}+1}+\alpha_2^{2^\frac{h}{2}+1}+\alpha_3^{2^\frac{h}{2}+1}=0$.

    \item $\textbf{Case 2:}$ When $p$ is odd and $q=p^h>4$, we have $(p^h-1)\nmid (p^\frac{h}{2}+1)$, and thus, we can choose 
    $\alpha_1 \in \mathbb{F}_{p^h}^*$ satisfying $\alpha_1^{p^{\frac{h}{2}}+1}\neq 1$. It follows from    
    \begin{align*}
        (1-\alpha_1^{p^{\frac{h}{2}}+1})^{p^{\frac{h}{2}}}=1-\alpha_1^{p^{\frac{h}{2}}+1}
    \end{align*}
    that $1-\alpha_1^{p^{\frac{h}{2}}+1}\in \mathbb{F}_{p^\frac{h}{2}}^*$. Hence, there is an $\alpha_2 \in \mathbb{F}_{p^h}^*$ such that 
    $\alpha_2^{p^{\frac{h}{2}}+1}=1-\alpha_1^{p^{\frac{h}{2}}+1}$, i.e., $\alpha_1^{p^{\frac{h}{2}}+1}+\alpha_2^{p^{\frac{h}{2}}+1}=1$. 
    Note that the Condition (3) in Theorem \ref{th.alpha_1,alpha_2,...,alpha_{2i} exist} holds, then by Lemma \ref{lem.t_0 exists}, 
    there is an integer $t_0$ such that $\omega^{t_0(p^\frac{h}{2}+1)}=-1$, where $\omega$ is a primitive element 
    of $\mathbb{F}_{p^h}$. Denote $\alpha_3=\omega^{t_0}\in \mathbb{F}_{p^h}^*$, then 
    $\alpha_1^{p^{\frac{h}{2}}+1}+\alpha_2^{p^{\frac{h}{2}}+1}+\alpha_3^{p^{\frac{h}{2}}+1}=1-1=0$.
\end{itemize}

From Theorem \ref{th.alpha_1,alpha_2,...,alpha_{2i} exist} (2), we know that there exists an $[n+2i+1,k,\widetilde{d}_{2i+1}]_q$ 
Hermitian self-orthogonal code $\mathcal{C}_{2i+1}$ for each $i\geq 1$, where $d\leq \widetilde{d}_{2i+1}\leq n+2(i+1)-k$, which completes the proof. 
\end{proof}

Noting that the results (1) and (2) of Lemma \ref{Coro.Euclidean} (resp. Lemma \ref{Coro.Hermitian}) have similar conditions and conclusions, 
we can rewrite Lemma \ref{Coro.Euclidean} (resp. Lemma \ref{Coro.Hermitian}) equivalently as Theorem \ref{th.Euclidean} (resp. Theorem \ref{th.Hermitian}) 
in order to unify it with the form of Theorem \ref{eq.th.n+i}.

\begin{theorem}\label{th.Euclidean}
    Let $\mathcal{C}$ be an $[n,k,d]_q$ Euclidean self-orthogonal code. For any integer $i\geq 0$ and $i\neq 1$, if $p$ is even and $h\geq 2$, or $p$ is odd and $h$ is even, 
    there exists an $[n+i,k,\widetilde{d}_{i}]_q$ Euclidean self-orthogonal code $\mathcal{C}_{i}$, where $d\leq \widetilde{d}_{i}\leq n+i+1-k$. 
    Moreover, for $q=2$, there exists an $[n+2i,k,\widetilde{d}_{2i}]_2$ Euclidean self-orthogonal code $\mathcal{C}_{2i}$ for each $i\geq 0$, where $d\leq \widetilde{d}_{2i}\leq n+2i+1-k$.  
\end{theorem}

\begin{theorem}\label{th.Hermitian}
    Let $\mathcal{C}$ be an $[n,k,d]_q$ Hermitian self-orthogonal code. For any integer $i\geq 0$ and $i\neq 1$, if $q>4$, there exists an $[n+i,k,\widetilde{d}_{i}]_q$ 
    Hermitian self-orthogonal code $\mathcal{C}_{i}$, where $d\leq \widetilde{d}_{i}\leq n+i+1-k$. 
    Moreover, for $q=4$, there exists an $[n+2i,k,\widetilde{d}_{2i}]_4$ Hermitian self-orthogonal code $\mathcal{C}_{2i}$ for each $i\geq 0$, where $d\leq \widetilde{d}_{2i}\leq n+2i+1-k$.  
\end{theorem}

\subsection{Results on codes with Galois hulls of arbitrary dimensions}\label{sec4.4}
Recall that our ultimate goal in this section is to construct codes of larger length with Galois hulls of arbitrary dimensions from given 
Galois self-orthogonal codes. Based on the results of the first three subsections, we can achieve it in this subsection. 

Firstly, from Lemma \ref{coro.all_Galois hull from Galois self-orthogonal code}, codes of length $n+i$ ($i\geq 0$ and $i\neq 1$) with Galois hulls of 
arbitrary dimensions can be directly obtained for $q\geq 3$. We list them as follows. Specifically, 
Corollary \ref{coro.new_arbitrary dimension e-Galois hull codes with length n+2i and n+2i+1} comes from Theorem \ref{th.n+2i}, 
Corollary \ref{coro.new_arbitrary dimension Euclidean hull codes with length n+2i and n+2i+1} comes from Theorem \ref{th.Euclidean}, and 
Corollary \ref{coro.new_arbitrary dimension Hermitian hull codes with length n+2i and n+2i+1} comes from Theorem \ref{th.Hermitian}. 

\begin{corollary}\label{coro.new_arbitrary dimension e-Galois hull codes with length n+2i and n+2i+1}
    Let $q\geq 3$ and $\mathcal{C}$ be an $[n,k,d]_q$ $e$-Galois self-orthogonal code. If Equation (\ref{eq.th.n+i}) holds, 
    there exists an $[n+i,k,\widetilde{d}_{i}]_q$ code $\mathcal{C}_{i}$ with $l$-dimensional $e$-Galois hull for each $i\geq 0$, $i\neq 1$ and $0\leq l\leq k$, 
    where $d\leq \widetilde{d}_{i}\leq n+i+1-k$. 

\end{corollary}

\begin{corollary}\label{coro.new_arbitrary dimension Euclidean hull codes with length n+2i and n+2i+1}
    Let $q\geq 3$ and $\mathcal{C}$ be an $[n,k,d]_q$ Euclidean self-orthogonal code. If $p$ is even and $h\geq 2$, or $p$ is odd and $h$ is even, 
    there exists an $[n+i,k,\widetilde{d}_{i}]_q$ code $\mathcal{C}_{i}$ with $l$-dimensional Euclidean hull for each $i\geq 0$, $i\neq 1$ and $0\leq l\leq k$, 
    where $d\leq \widetilde{d}_{i}\leq n+i+1-k$. 
\end{corollary}

\begin{corollary}\label{coro.new_arbitrary dimension Hermitian hull codes with length n+2i and n+2i+1}
    Let $\mathcal{C}$ be an $[n,k,d]_q$ Hermitian self-orthogonal code. For $q> 4$, there exists an $[n+i,k,\widetilde{d}_{i}]_q$ code $\mathcal{C}_{i}$ 
    with $l$-dimensional Hermitian hull for each $i\geq 0$, $i\neq 1$ and $0\leq l\leq k$, where $d\leq \widetilde{d}_{i}\leq n+i+1-k$. Moreover, for $q=4$, 
    there exists an $[n+2i,k,\widetilde{d}_{2i}]_q$ code $\mathcal{C}_{2i}$ with $l$-dimensional Hermitian hull for each $i\geq 0$ and $0\leq l\leq k$, 
    where $d\leq \widetilde{d}_{2i}\leq n+2i+1-k$. 
\end{corollary}

As shown in Remark \ref{rem5.th.n+2i} (2), $i=1$ is not contained in Theorem \ref{th.n+2i}. 
Therefore, codes of length $n+1$ with Galois hulls of arbitrary dimensions can not be obtained from the above three corollaries. 
For this case, we give the following theorem, which shows some necessary and sufficient 
conditions to construct codes of length $n+1$ with $l$-dimensional $e$-Galois hull for prescribed $0\leq l\leq k-1$.

\begin{theorem}\label{th.arbitrary hull codess with length n+1}
    Let $\mathcal{C}$ be an $[n,k,d]_q$ $e$-Galois self-orthogonal code. Then the following statements hold. 
    \begin{itemize}
        \item [\rm (1)] There exists an $[n+1,k,\widetilde{d}_1]_q$ code $\C_1$ with $(k-1)$-dimensional $e$-Galois hull, 
        where $d\leq \widetilde{d}_1\leq n+2-k$;

        \item [\rm (2)] There exists an $[n+1,k,\widetilde{d}_1]_q$ code $\C_1$ with $l$-dimensional $e$-Galois hull, 
        where $d\leq \widetilde{d}_1\leq n+2-k$ if and only if $(k-l)\alpha^{p^e+1}+\beta^{p^e+1}\neq 1$ and $\beta^{p^e+1}\neq 1$ 
        hold for some $\alpha,\beta\in \mathbb{F}_q^*$ and prescribed $0\leq l\leq k-2$;

        \item [\rm (3)] In particular, if $\C$ is a GRS code of length $n\leq q$, then $\C_1$ derived in $(1)$ above is MDS, i.e., $\widetilde{d}_1=d+1$.
    \end{itemize}
\end{theorem}
\begin{proof}
    Similar to Theorem \ref{th.n+2i}, let $G=(I_k\ |\ A)$ be a generator matrix of $\mathcal{C}$, then $A\sigma^{e}(A^T)=-I_k$. 
Choose $\alpha,\beta,\lambda\in \mathbb{F}_q^*$ satisfying $\lambda^{p^e+1}=1$. Let    
\begin{equation*}
    G_{1}= \left(                 
         \begin{array}{c|c|c}   
            G' & A & \mathbf{a}
         \end{array} 
     \right)               
\end{equation*}
be a generator matrix of a code $\mathcal{C}_{1}$, where 
$G'={\rm diag}(\underbrace{\beta,\ \cdots,\ \beta}_{k-l},\ \underbrace{\lambda,\ \cdots,\ \lambda}_{l})$ and  
$\mathbf{a}=(\underbrace{\alpha\ \cdots\ \alpha}_{k-l}\ \underbrace{ 0\ \cdots\ 0}_{l})^T$. Then $\dim(\C_1)=\dim(\C)=k$. 
According to Remark \ref{rem5.th.n+2i} (1), adding $\mathbf{a}$ can make sense only if $k-l>0$, i.e., $0\leq l\leq k-1$. Note that 
    \begin{align*}
        G_{1}\sigma^{e}(G_{1}^T) & =  \left(                 
        \begin{array}{c|c|c}   
           G' & A & \mathbf{a}
        \end{array} 
    \right)
    \left(                 
        \begin{array}{c}   
           \sigma^e(G'^T) \\ \hline 
           \sigma^e(A^T) \\ \hline
           \sigma^e(\mathbf{a}^T) \\
        \end{array}
    \right)\\
    &  =  G'\sigma^e(G'^T)+A\sigma^e(A^T)+\mathbf{a}\sigma^e(\mathbf{a}^T) \\
    &  =  
	\left( 
	\begin{array}{c|c}
            \begin{array}{cccc}
            \alpha^{p^e+1}+\beta^{p^e+1}-1 & \alpha^{p^e+1}                 &\cdots & \alpha^{p^e+1}\\
            \alpha^{p^e+1}                 & \alpha^{p^e+1}+\beta^{p^e+1}-1 &\cdots & \alpha^{p^e+1}\\
            \vdots                         & \vdots                              &\ddots & \vdots \\
            \alpha^{p^e+1}                 &\alpha^{p^e+1}                  &\cdots & \alpha^{p^e+1}+\beta^{p^e+1}-1
            \end{array}
                                                                                                                            & O_{(k-l)\times l} \\ \hline 
	O_{l\times (k-l)}                                                                                                       & O_{l\times l}
	\end{array}
	\right).
    \end{align*}

    Denote by $D_{(k-l)\times (k-l)}$ the upper left corner of the above block matrix. 
    From Equation (\ref{eq.panxu hull}), we know that 
    \begin{align*}
        \begin{split}
            \dim(\Hull_e(\C_1))=l & \Leftrightarrow \rank(G_{1}\sigma^{e}(G_{1}^T))=\dim(\C_1)-l \\
                                & \Leftrightarrow \rank(D_{(k-l)\times (k-l)})=k-l \\
                                & \Leftrightarrow  D_{(k-l)\times (k-l)} {\rm~is~nonsingular}.
        \end{split}
    \end{align*}
    For $0\leq l\leq k-1$, we now determine whether $D_{(k-l)\times (k-l)}$ is a nonsingular matrix or not.  
    An easy calculation follows that 
    \begin{align}\label{eq.det(D)}
        {\rm det}(D_{(k-l)\times (k-l)})=(\beta^{p^e+1}-1)^{k-l-1}\cdot ((k-l)\alpha^{p^e+1}+\beta^{p^e+1}-1),        
    \end{align}
    where ${\rm det}(D_{(k-l)\times (k-l)})$ denotes the determinant of the matrix $D_{(k-l)\times (k-l)}$.
    
    We discuss it in the following two cases. 
    \begin{itemize}
        \item $\textbf{Case 1:}$ When $l=k-1$, i.e., $k-l=1$, by Equation (\ref{eq.det(D)}), we have 
        ${\rm det}(D_{1\times 1})=\alpha^{p^e+1}+\beta^{p^e+1}-1$. Clearly, 
        ${\rm det}(D_{1\times 1})\neq 0$ if and only if $\alpha^{p^e+1}+\beta^{p^e+1}\neq 1$ holds 
        for some $\alpha,\ \beta\in \mathbb{F}_q^*$. We assert that such $\alpha,\beta$ always exist. 
        Otherwise, if we take $\beta=1$, then $\alpha^{p^e+1}=0$, which contradicts to $\alpha\in \F_q^*$. 
        Therefore, $D_{1\times 1}$ is nonsingular for any $\alpha,\beta\in \mathbb{F}_q^*$ and in this case, 
        $\dim(\Hull_e(\C_1))=k-1$. 
        
        \item $\textbf{Case 2:}$ When $0\leq l\leq k-2$, by Equation (\ref{eq.det(D)}), we have 
        ${\rm det}(D_{(k-l)\times (k-l)})\neq 0$ if and only if $\beta^{p^e+1}\neq 1$ and 
        $(k-l)\alpha^{p^e+1}+\beta^{p^e+1}\neq 1$ hold for some $\alpha,\beta\in \mathbb{F}_q^*$.   
        And in this case, $\dim(\Hull_e(\C_1))=l$.  
    \end{itemize}
    Note that the range of $\widetilde{d}_1$ can be deduced with the same reasoning as Theorem \ref{th.n+2i}. 
    Therefore, we complete the proofs of the results $(1)$ and $(2)$. 
  
    For the proof of the result (3), we assume that $\mathcal{C}$ is an $e$-Galois self-orthogonal GRS code of length $n\leq q$ with a generator matrix $G$, 
    then $G\sigma^e(G^T)=O_{k\times k}$. Take $G_1 =(G\ |\ \mathbf{g})$ as a generator matrix of a code $\mathcal{C}_1'$, 
    where $\mathbf{g}=(\underbrace{0\ 0\ \cdots\ 0}_{k-1}\ \gamma)^T$ with $\gamma\in \mathbb{F}_q^*$. 
    It follows from
    \begin{align*}
        G_1\sigma^e(G_1^T)=G\sigma^e(G^T)+\mathbf{g}\sigma^e(\mathbf{g}^T)={\rm diag}(\underbrace{0,\ 0,\ \cdots,\ 0}_{k-1}, \gamma^{p^e+1})    
    \end{align*} 
    that $\dim(\Hull_e(\mathcal{C}_1))=k-\rank(G_1\sigma^e(G_1^T))=k-1$. 
    Note that when $\gamma = 1$, $G_1$ generates a so-called extended GRS code with the same code locators and column multipliers as $\C$, 
    which is MDS. 
    Therefore, for $\gamma\in \F_q^*$, 
    $G_1$ generates a code equivalent to the extended GRS code. Recall that equivalent codes have the same minimum distance, 
    and thus, $\mathcal{C}_1'$ is an $[n+1,k,d+1]_q$ MDS code with $(k-1)$-dimensional $e$-Galois hull, which completes the proof of the result $(3)$. 
\end{proof}

\begin{remark}\label{rem8.length of n+1} $\quad$
    \begin{enumerate}
        \item [\rm (1)] In Theorem \ref{th.arbitrary hull codess with length n+1}, adding one column can make sense 
        only if $0\leq l\leq k-1$. Therefore, we can not obtain $e$-Galois self-orthogonal codes of length $n+1$ 
        in Theorem \ref{th.arbitrary hull codess with length n+1}, and thus, the study of Theorem \ref{th.n+2i} is necessary. 

        \item [\rm (2)]  For more details on extended GRS codes, readers are referred to \cite{RefJ-6}-\cite{RefJ-4} and the references therein. 
    \end{enumerate}
\end{remark}

\subsection{Examples}\label{sec4.5-Examples}

In this subsection, we give some interesting examples to illustrate our results.
To this end, we first introduce some useful definitions. 

\begin{definition}{\rm (\cite{RefJ2})}\label{def.Singleton detect}
   The \underline{Singleton defect} of an $[n,k,d]_q$ code $\mathcal{C}$ is $S(\mathcal{C})=n-k+1-d.$
\end{definition}

Based on Definition \ref{def.Singleton detect}, Liao et al. \cite{RefJ1} introduced the so-called $m$-MDS code under the  
Euclidean inner product. Here, we generalize it into the general $e$-Galois inner product in a natural way as follows. 

\begin{definition}\label{def.m-MDS codes}
    An $[n,k,d]_q$ code $\mathcal{C}$ is called \underline{$e$-Galois $m$-MDS}  if $S(\mathcal{C})=S(\mathcal{C}^{\bot_e})=m.$
\end{definition}
\begin{remark}\label{rem5} $\quad$ 
    \begin{enumerate}
        \item [\rm (1)] When $e=0$, it is the origin definition proposed by Liao et al. \cite{RefJ1}. 
        \item [\rm (2)] When $m=0$ (resp. $m=1$), $\mathcal{C}$ is the important MDS (resp. NMDS) code.
        \item [\rm (3)] As stated in \cite{RefJ0}, $m$-MDS codes can make sense if 
        $0\leq m\leq \min\{k-1,\lfloor \frac{n}{2}\rfloor\}$. In \cite{RefJ1}, some properties of $m$-MDS codes 
        are discussed and some $m$-MDS codes were constructed via Reed-Muller codes and Golay codes.
    \end{enumerate}
\end{remark}

\begin{definition}\label{def.optimal codes}
    Let $\mathcal{C}$ be an $[n,k,d]_q$ code. 
    \begin{itemize}
        \item [\rm (1)] $\C$ is called \underline{distance-optimal} if $\C$ has the largest minimum distance among all $[n,k]_q$ codes; 
        \item [\rm (2)] $\C$ is called \underline{almost-distance-optimal} if there exists an $[n,k,d+1]_q$ distance-optimal code.  
    \end{itemize} 
\end{definition}

Due to space limitations, we mainly give some examples on $e$-Galois self-orthogonal codes of length $n+2$ derived from 
$e$-Galois self-orthogonal codes of length $n$, and more other lengths can be derived similarly. Notably, in our examples, 
most derived codes are distance-optimal, almost-distance-optimal, or $m$-MDS, and their minimum distances 
are improved by different degrees. 

\begin{example}\label{example1}
    Let $p=2$, $h=1$, i.e., $q=2$. In the first two examples, we consider the Euclidean inner product. 
    \begin{enumerate}
        \item [\rm (1)] Let  
        \begin{footnotesize}
            \begin{align*}
                G= \left( 
                    \begin{array}{cccc}
                        1 & 0 & 1 & 0 \\
                        0 & 1 & 0 & 1 \\
                    \end{array}
                    \right)    
            \end{align*}
        \end{footnotesize}
            be a generator matrix of a $[4,2,2]_2$ Euclidean self-dual code $\mathcal{C}$. 
            Since $p$ is even and $e=0$ here, according to Theorem \ref{th.alpha_1,alpha_2,...,alpha_{2i} exist}, 
            we can take $\alpha_1=\alpha_2=1\in \mathbb{F}_2^*$, then $G_2$ in Theorem \ref{th.n+2i} can be written as 
        \begin{footnotesize}
                \begin{align*}
                    G_2= \left( 
                        \begin{array}{cccccc}
                            1 & 0 & 1 & 0 & 1 & 1\\
                            0 & 1 & 0 & 1 & 1 & 1\\
                        \end{array}
                        \right).    
                \end{align*}
        \end{footnotesize}
    
           $\quad$ Computed with the Magma software package \cite{BCP1997}, we know that the code $\mathcal{C}_2$ generated by $G_2$ is a $[6,2,4]_2$ 
           Euclidean self-orthogonal code and its Euclidean dual code $\mathcal{C}_2^{\bot_0}$ has parameters 
           $[6,4,2]_2$. On one hand, according to Codetable \cite{RefJ-1}, both $\C_2$ and $\C_2^{\bot_0}$ are distance-optimal. On the other hand, 
           since $S(\mathcal{C}_2)=S(\mathcal{C}_2^{\bot_0})=1\leq \min\{1,3\}$, $\mathcal{C}_2$ is an Euclidean self-orthogonal NMDS code 
           over $\mathbb{F}_2$. Moreover, we notice that both the code length and the minimum distance increase $2$ here.

           \item [\rm (2)] Let 
           \begin{footnotesize}
              \begin{align*}
                  G= \left( 
                      \begin{array}{cccccc}
                          1 & 0 & 0 & 1 & 0 & 0\\
                          0 & 1 & 0 & 0 & 1 & 0 \\
                          0 & 0 & 1 & 0 & 0 & 1  \\
                      \end{array}
                      \right).
              \end{align*}
           \end{footnotesize}
           be a generator matrix of a $[6,3,2]_2$ Euclidean self-dual code $\mathcal{C}$. 
           According to Codetable \cite{RefJ-1}, we notice that $\C$ is not a distance-optimal code.  
           Since $p$ is even and $e=0$ here, according to Theorem \ref{th.alpha_1,alpha_2,...,alpha_{2i} exist}, we can take 
           $\alpha_1=\alpha_2=1\in \mathbb{F}_{2}^*$, then $G_2$ in Theorem \ref{th.n+2i} can be written as 
           \begin{footnotesize}
                  \begin{align*}
                      G_2 = \left( 
                          \begin{array}{cccccccc}
                           1 & 0 & 0 & 1 & 0 & 0 & 1 & 1\\
                           0 & 1 & 0 & 0 & 1 & 0 & 1 & 1\\
                           0 & 0 & 1 & 0 & 0 & 1 & 1 & 1\\
                          \end{array}
                          \right).
                  \end{align*}
           \end{footnotesize}
       
              $\quad$ Computed with the Magma software package \cite{BCP1997}, we know that the code $\mathcal{C}_2$ generated by $G_2$ is an 
              $[8,3,4]_{2}$ Euclidean self-orthogonal code and its Euclidean dual code $\mathcal{C}_2^{\bot_0}$ has parameters 
              $[8,5,2]_{2}$. On one hand, according to Codetable \cite{RefJ-1}, both $\C_2$ and $\C_2^{\bot_0}$ are distance-optimal. 
              On the other hand, since $S(\mathcal{C}_2)=S(\mathcal{C}_2^{\bot_0})=2\leq \min\{2,4\}$, $\mathcal{C}_2$ is an Euclidean self-orthogonal $2$-MDS code 
              over $\mathbb{F}_2$. Moreover, we notice that both the code length and the minimum distance increase $2$ here.
    \end{enumerate}
\end{example}

\begin{example}\label{example2}
    Next, we give four examples under the Hermitian inner product. 
    \begin{enumerate}
        \item [\rm (1)] Let $p=2$, $h=2$, i.e., $q=4$ and $\omega$ be a primitive element of $\F_4$. Let  
        \begin{footnotesize}
            \begin{align*}
                G= \left( 
                    \begin{array}{cccc}
                        1 & 0 & 1 & 0 \\
                        0 & 1 & 0 & 1 \\
                    \end{array}
                    \right)    
            \end{align*}
        \end{footnotesize}
            be a generator matrix of a $[4,2,2]_4$ Hermitian self-dual code $\mathcal{C}$. 
            Since $p$ is even and $e=1$ here, according to Theorem \ref{th.alpha_1,alpha_2,...,alpha_{2i} exist}, 
            we can take $\alpha_1=\alpha_2=\omega\in \mathbb{F}_4^*$, then $G_2$ in Theorem \ref{th.n+2i} can be written as 
        \begin{footnotesize}
                \begin{align*}
                    G_2= \left( 
                        \begin{array}{cccccc}
                            1 & 0 & 1 & 0 & \omega & \omega\\
                            0 & 1 & 0 & 1 & \omega & \omega\\
                        \end{array}
                        \right).    
                \end{align*}
        \end{footnotesize}
    
        $\quad$ Computed with the Magma software package \cite{BCP1997}, we know that the code $\mathcal{C}_2$ generated by $G_2$ is a $[6,2,4]_4$ 
           Hermitian self-orthogonal code and its Hermitian dual code $\mathcal{C}_2^{\bot_1}$ has parameters 
           $[6,4,2]_4$. On one hand, according to Codetable \cite{RefJ-1}, both $\C_2$ and $\C_2^{\bot_1}$ are distance-optimal. On the other hand, 
           since $S(\mathcal{C}_2)=S(\mathcal{C}_2^{\bot_1})=1\leq \min\{1,3\}$, $\mathcal{C}_2$ is a Hermitian self-orthogonal NMDS code 
           over $\mathbb{F}_4$. And we also notice that both the code length and the minimum distance increase $2$ here. 

           \item [\rm (2)] Let notations be the same as {\rm (1)} above. Let 
           \begin{footnotesize}
           \begin{align*}
               G= \left( 
                   \begin{array}{ccccccccccc}
                       1 & 0 & 0 & 0 & 0 & 1 & 0 & 1 & \omega & 1 & \omega \\
                       0 & 1 & 0 & 0 & 0 & 0 & 0 & 1 & 0 & \omega^{2} & \omega^{2} \\
                       0 & 0 & 1 & 0 & 0 & \omega & \omega & 0 & 1 & \omega^{2} & \omega^{2} \\
                       0 & 0 & 0 & 1 & 0 & \omega & \omega & \omega & \omega & 1 & 0 \\
                       0 & 0 & 0 & 0 & 1 & 0 & 0 & \omega & \omega & 0 & 1 \\
                   \end{array}
                   \right)
           \end{align*}
           \end{footnotesize}
           be a generator matrix of a $[11,5,4]_4$ Herimitian self-orthogonal code $\C$. 
           Since $p$ is even and $e=1$ here, according to Theorem \ref{th.alpha_1,alpha_2,...,alpha_{2i} exist}, we can take $\alpha_1=\alpha_2=\omega^2\in \mathbb{F}_{4}^*$, 
           then $G_2$ in Theorem \ref{th.n+2i} can be written as 
       \begin{footnotesize}
           \begin{align*}
               G_2 = \left( 
                \begin{array}{ccccccccccccc}
                    1 & 0 & 0 & 0 & 0 & 1 & 0 & 1 & \omega & 1 & \omega & \omega^2 & \omega^2\\
                    0 & 1 & 0 & 0 & 0 & 0 & 0 & 1 & 0 & \omega^{2} & \omega^{2} & \omega^2 & \omega^2\\
                    0 & 0 & 1 & 0 & 0 & \omega & \omega & 0 & 1 & \omega^{2} & \omega^{2} & \omega^2 & \omega^2\\
                    0 & 0 & 0 & 1 & 0 & \omega & \omega & \omega & \omega & 1 & 0 & \omega^2 & \omega^2\\
                    0 & 0 & 0 & 0 & 1 & 0 & 0 & \omega & \omega & 0 & 1 & \omega^2 & \omega^2 \\
                \end{array}
                \right).
           \end{align*}
       \end{footnotesize}
       
       $\quad$ Computed with the Magma software package \cite{BCP1997}, we know that the code $\mathcal{C}_2$ 
           generated by $G_2$ is a $[13,5,6]_{4}$ Hermitian self-orthogonal code. Note that $\C$ is not 
           distance-optimal or almost-distance-optimal, while the new code $\C_2$ is almost-distance-optimal \cite{RefJ-1}. 
           Moreover, as the code length increases by $2$, the minimum distance also increases by $2$ in this example. 

           \item [\rm (3)] Let $p=3$, $h=2$, i.e., $q=9$ and $\omega$ be a primitive element of $\mathbb{F}_{9}$. Let 
           \begin{footnotesize}
           \begin{align*}
               G= \left( 
                   \begin{array}{ccccccccccccc}
                       1 & 0 & \omega^{6} & \omega & 2 & \omega^{6} & \omega^{3} & \omega^{5} & \omega^{2} & \omega^{3} & 0 & \omega^{2} & 1 \\
                       0 & 1 & 1 & \omega^{2} & \omega^{3} & \omega^{3} & \omega^{7} & \omega^{3} & \omega^{3} & \omega^{5} & \omega^{3} & 0 & \omega^{6}\\
                   \end{array}
                   \right)
           \end{align*}
           \end{footnotesize}
           be a generator matrix of a $[13,2,11]_9$ Hermitian self-orthogonal code $\C$ whose Hermitian dual code has parameters $[13,11,2]_9$. 
           Since $p$ is odd and $e=1$ here, according to Theorem \ref{th.alpha_1,alpha_2,...,alpha_{2i} exist}, 
           we can take $\alpha_1=2$, $\alpha_2=\omega\in \mathbb{F}_{9}^*$, then $G_2$ in Theorem \ref{th.n+2i} can be written as 
           \begin{footnotesize}
            \begin{align*}
                G= \left( 
                    \begin{array}{ccccccccccccccc}
                        1 & 0 & \omega^{6} & \omega & 2 & \omega^{6} & \omega^{3} & \omega^{5} & \omega^{2} & \omega^{3} & 0 & \omega^{2} & 1 & 2 & \omega \\
                        0 & 1 & 1 & \omega^{2} & \omega^{3} & \omega^{3} & \omega^{7} & \omega^{3} & \omega^{3} & \omega^{5} & \omega^{3} & 0 & \omega^{6} & 2 & \omega \\
                    \end{array}
                    \right). 
            \end{align*}
            \end{footnotesize}
        
            $\quad$ Computed with the Magma software package \cite{BCP1997}, we know that the code $\mathcal{C}_2$ 
           generated by $G_2$ is a $[15,2,13]_{9}$ Hermitian self-orthogonal code and its Hermitian dual code 
           $\mathcal{C}_2^{\bot_1}$ has parameters $[15,13,2]_{9}$. On one hand, according to Codetable \cite{RefJ-1}, both $\C_2$ and $\C_2^{\bot_1}$ are 
           distance-optimal. On the other hand, since $S(\mathcal{C}_2)=S(\mathcal{C}_2^{\bot_2})=1\leq \min\{1,7\}$, 
            $\mathcal{C}_2$ is a Hermitian self-orthogonal NMDS code. 

           \item [\rm (4)] Let $p=3$, $h=4$, i.e., $q=3^4$ and $\omega$ be a primitive element of $\mathbb{F}_{3^4}$. 
           From \cite[Example 5]{RefJ B3}, we can get a $[17,3,15]_{3^4}$ Hermitian self-orthogonal code $\mathcal{C}$ 
           with a generator matrix 
           \begin{footnotesize}
           \begin{align*}
               G= \left( 
                   \begin{array}{ccccccccccccccccc}
                       1 & 0 & 0 & \omega^{53} & \omega^{11} & \omega^{31} & \omega^{72} & \omega^{32} & \omega^{20} & \omega^{50} & \omega^{55} & \omega^{35} & \omega^{57} & \omega^{27} & \omega^{76} & \omega^{66} & \omega^{38} \\
                       0 & 1 & 0 & \omega^{55} & \omega^{71} & \omega^{23} & \omega^{21} & \omega^{74} & \omega^{71} & \omega^{72} & \omega^{31} & \omega^{7}  & \omega^{21} & \omega^{39} & \omega      & \omega^{18} & \omega \\
                       0 & 0 & 1 & \omega^{16} & \omega^{48} & \omega^{46} & \omega^{19} & \omega^{16} & \omega^{17} & \omega^{56} & \omega^{32} & \omega^{46} & \omega^{64} & \omega^{26} & \omega^{43} & \omega^{26} & \omega^{25}\\
                   \end{array}
                   \right).
           \end{align*}
           \end{footnotesize}
           Since $p$ is odd and $e=2$ here, according to Theorem \ref{th.alpha_1,alpha_2,...,alpha_{2i} exist}, we can take $\alpha_1=\omega$, $\alpha_2=\omega^5\in \mathbb{F}_{3^4}^*$, then $G_2$ in Theorem \ref{th.n+2i} can be written as 
       \begin{footnotesize}
           \begin{align*}
               G_2 = \left( 
                   \begin{array}{ccccccccccccccccccc}
                       1 & 0 & 0 & \omega^{53} & \omega^{11} & \omega^{31} & \omega^{72} & \omega^{32} & \omega^{20} & \omega^{50} & \omega^{55} & \omega^{35} & \omega^{57} & \omega^{27} & \omega^{76} & \omega^{66} & \omega^{38} & \omega & \omega^5\\
                       0 & 1 & 0 & \omega^{55} & \omega^{71} & \omega^{23} & \omega^{21} & \omega^{74} & \omega^{71} & \omega^{72} & \omega^{31} & \omega^{7}  & \omega^{21} & \omega^{39} & \omega      & \omega^{18} & \omega & \omega & \omega^5\\
                       0 & 0 & 1 & \omega^{16} & \omega^{48} & \omega^{46} & \omega^{19} & \omega^{16} & \omega^{17} & \omega^{56} & \omega^{32} & \omega^{46} & \omega^{64} & \omega^{26} & \omega^{43} & \omega^{26} & \omega^{25} & \omega & \omega^5\\
                   \end{array}
                   \right).
           \end{align*}
       \end{footnotesize}
       
       $\quad$ Computed with the Magma software package \cite{BCP1997}, we know that the code $\mathcal{C}_2$ 
           generated by $G_2$ is a $[19,3,15]_{3^4}$ Hermitian self-orthogonal code and its Hermitian dual code 
           $\mathcal{C}_2^{\bot_2}$ has parameters $[19,16,2]_{3^4}$. 
           Since $S(\mathcal{C}_2)=S(\mathcal{C}_2^{\bot_2})=2\leq \min\{2,9\}$, 
           $\mathcal{C}_2$ is a Hermitian self-orthogonal $2$-MDS code. 
    \end{enumerate}

\end{example}

\begin{example}\label{example3}
    And then, we give some examples based on $1$-Galois self-orthogonal constacyclic codes as follows. 
        \begin{enumerate}
            \item [\rm (1)] Let $p=2$, $h=6$, i.e., $q=2^6$ and $\omega$ be a primitive element of $\mathbb{F}_{2^6}$. From  
            \cite[Example 6.1]{RefJ52}, taking $\lambda=w^{21}$, $n=18$ and the generator polynomial $g(x)=x^{12}+w^7x^6+w^{14}$, we can get 
            an $[18,6,3]_{2^6}$ $1$-Galois (resp. $5$-Galois according to Corollary \ref{coro. Galois self-orthogonal code judgment}) 
            self-orthogonal $w^{21}$-constacyclic code $\mathcal{C}$ with a generator matrix 
            \begin{footnotesize}
            \begin{align*}
                G= \left( 
                    \begin{array}{cccccccccccccccccc}
                        1 & 0 & 0 & 0 & 0 & 0 & \omega^{14} & 0 & 0 &0 & 0 & 0 & \omega^{28} & 0 & 0 & 0 &0 & 0 \\
                        0 & 1 & 0 & 0 & 0 & 0 & 0 & \omega^{14} & 0 &0 & 0 & 0 & 0 & \omega^{28} & 0 & 0 &0 & 0 \\
                        0 & 0 & 1 & 0 & 0 & 0 & 0 & 0 & \omega^{14} & 0 & 0 & 0 & 0 & 0 & \omega^{28} & 0 &0 & 0 \\
                        0 & 0 & 0 & 1 & 0 & 0 & 0 & 0 & 0 & \omega^{14} & 0 & 0 & 0 & 0 & 0 & \omega^{28} & 0 & 0 \\
                        0 & 0 & 0 & 0 & 1 & 0 & 0 & 0 & 0 &0 & \omega^{14} & 0 & 0 & 0 & 0 & 0 & \omega^{28} & 0 \\
                        0 & 0 & 0 & 0 & 0 & 1 & 0 & 0 & 0 &0 & 0 & \omega^{14} & 0 & 0 & 0 & 0 & 0 & \omega^{28} \\
                    \end{array}
                    \right).
            \end{align*}
        \end{footnotesize}
            Since $p$ is even and $e=1$ here, according to Theorem \ref{th.alpha_1,alpha_2,...,alpha_{2i} exist}, we can take 
            $\alpha_1=\alpha_2=\omega^{5}\in \mathbb{F}_{2^6}^*$, then $G_2$ in Theorem \ref{th.n+2i} can be written as 
            \begin{footnotesize}
                \begin{align*}
                    G_2 = \left( 
                        \begin{array}{cccccccccccccccccccc}
                            1 & 0 & 0 & 0 & 0 & 0 & \omega^{14} & 0 & 0 &0 & 0 & 0 & \omega^{28} & 0 & 0 & 0 &0 & 0 & \omega^5 & \omega^5 \\
                            0 & 1 & 0 & 0 & 0 & 0 & 0 & \omega^{14} & 0 &0 & 0 & 0 & 0 & \omega^{28} & 0 & 0 &0 & 0 & \omega^5 & \omega^5 \\
                            0 & 0 & 1 & 0 & 0 & 0 & 0 & 0 & \omega^{14} & 0 & 0 & 0 & 0 & 0 & \omega^{28} & 0 &0 & 0 & \omega^5 & \omega^5 \\
                            0 & 0 & 0 & 1 & 0 & 0 & 0 & 0 & 0 & \omega^{14} & 0 & 0 & 0 & 0 & 0 & \omega^{28} & 0 & 0 & \omega^5 & \omega^5 \\
                            0 & 0 & 0 & 0 & 1 & 0 & 0 & 0 & 0 &0 & \omega^{14} & 0 & 0 & 0 & 0 & 0 & \omega^{28} & 0 & \omega^5 & \omega^5 \\
                            0 & 0 & 0 & 0 & 0 & 1 & 0 & 0 & 0 &0 & 0 & \omega^{14} & 0 & 0 & 0 & 0 & 0 & \omega^{28} & \omega^5 & \omega^5 \\
                        \end{array}
                        \right).
                \end{align*}
            \end{footnotesize}
        
            $\quad$ Computed with the Magma software package \cite{BCP1997}, we know that the code $\mathcal{C}_2$ generated by $G_2$ is a 
            $[20,6,5]_{2^6}$ $1$-Galois (resp. $5$-Galois) self-orthogonal code, whose code length and minimum distance both increase $2$. 
            
            $\quad$ Additionally, taking $\alpha_1=w^{50}, \alpha_2=\omega^{36}, \alpha_3=\omega^{22}\in \mathbb{F}_{2^6}^*$, $G_3$ in 
            Theorem \ref{th.n+2i} can be written as $G_3=\left(I_6\ |\ B\right)$, where $B$ is defined to be 
            \begin{footnotesize}
                \begin{align*}
                    B  = \left( 
                        \begin{array}{ccccccccccccccc}
                             \omega^{14} & 0 & 0 &0 & 0 & 0 & \omega^{28} & 0 & 0 & 0 &0 & 0 & \omega^{50} & \omega^{36} & \omega^{22} \\
                             0 & \omega^{14} & 0 &0 & 0 & 0 & 0 & \omega^{28} & 0 & 0 &0 & 0 & \omega^{50} & \omega^{36} & \omega^{22} \\
                             0 & 0 & \omega^{14} & 0 & 0 & 0 & 0 & 0 & \omega^{28} & 0 &0 & 0 & \omega^{50} & \omega^{36} & \omega^{22} \\
                             0 & 0 & 0 & \omega^{14} & 0 & 0 & 0 & 0 & 0 & \omega^{28} & 0 & 0 & \omega^{50} & \omega^{36} & \omega^{22} \\
                             0 & 0 & 0 &0 & \omega^{14} & 0 & 0 & 0 & 0 & 0 & \omega^{28} & 0 & \omega^{50} & \omega^{36} & \omega^{22} \\
                             0 & 0 & 0 &0 & 0 & \omega^{14} & 0 & 0 & 0 & 0 & 0 & \omega^{28} & \omega^{50} & \omega^{36} & \omega^{22} \\
                        \end{array}
                        \right).
                \end{align*}
            \end{footnotesize}
        
            $\quad$ Computed with the Magma software package \cite{BCP1997}, we know that the code $\mathcal{C}_3$ generated by $G_3$ is a 
            $[21,6,6]_{2^6}$ $1$-Galois (resp. $5$-Galois) self-orthogonal code, whose code length and minimum distance both increase $3$.

            \item [\rm (2)] Let notations be the same as Example \ref{example.333}. 
            Since $p$ is odd and $e=1$ here, according to Theorem \ref{th.alpha_1,alpha_2,...,alpha_{2i} exist}, we can take 
            $\alpha_1=2, \alpha_2=\omega^{10}\in \mathbb{F}_{3^4}^*$, then $G_2$ in Theorem \ref{th.n+2i} can be written as 
            \begin{footnotesize}
            \begin{align*}
                G_2 = \left( 
                    \begin{array}{cccccccccccccc}
                        1 & 0 & 0 & 0 & 0 & 0 & \omega^{50} & 0 & 1 &0 & \omega^{70} & 0 & 2 & \omega^{10} \\
                        0 & 1 & 0 & 0 & 0 & 0 & 0 & \omega^{50} & 0 & 1 & 0 & \omega^{70} & 2 & \omega^{10} \\
                        0 & 0 & 1 & 0 & 0 & 0 & \omega^{60} & 0 & 0 & 0 & 2 & 0 & 2 & \omega^{10} \\
                        0 & 0 & 0 & 1 & 0 & 0 & 0 & \omega^{60} & 0 & 0 & 0 & 2 & 2 & \omega^{10} \\
                        0 & 0 & 0 & 0 & 1 & 0 & \omega^{30} & 0 & \omega^{20} & 0 & \omega^{50} & 0 & 2 & \omega^{10} \\
                        0 & 0 & 0 & 0 & 0 & 1 & 0 & \omega^{30} & 0 & \omega^{20} & 0 & \omega^{50} & 2 & \omega^{10} \\
                    \end{array}
                    \right).
            \end{align*}
        \end{footnotesize}
        
        $\quad$ Computed with the Magma software package \cite{BCP1997}, we know that the code $\mathcal{C}_2$ generated by $G_2$ is a 
            $[14,6,4]_{3^4}$ $1$-Galois (resp. $3$-Galois) self-orthogonal code, whose minimum distance increases $1$ and 
            its $1$-Galois dual code $\mathcal{C}_2^{\bot_1}$ is a $[14,8,2]_{3^4}$ code. Since  
            $S(\mathcal{C}_2)=S(\mathcal{C}_2^{\bot_1})=5\leq \min\{5,7\}$,  $\mathcal{C}_2$ is a 
            $1$-Galois (resp. $3$-Galois) self-orthogonal $5$-MDS code.
        \end{enumerate}
\end{example}

\begin{example}\label{example4}
    At the end of this subsection, we present two examples for Theorem \ref{th.arbitrary hull codess with length n+1}. 
    \begin{enumerate}
        \item [\rm (1)] Let $p=3$, $h=2$, i.e., $q=9$ and $\omega$ be a primitive element of $\mathbb{F}_{9}$. Let 
        \begin{footnotesize}
        \begin{align*}
            G= \left( 
                \begin{array}{cccccccccccccc}
                    1 & 0 & 0 & 0 & \omega^{3} & \omega^3 & \omega^3 & 1 & 2 & \omega^{5} & 0 & 2 & \omega^{3} & \omega^6 \\
                    0 & 1 & 0 & 0 & \omega^{2} & 0 & \omega^{7} & \omega^{7} & \omega^2 & \omega & 2 & \omega^{5} & \omega^6 & 1  \\
                    0 & 0 & 1 & 0 & 2 & \omega^{3} & \omega^{5} & 2 & 0 & \omega & \omega^{5} & \omega^2 & 2 & \omega  \\
                    0 & 0 & 0 & 1 & \omega^{5} & \omega & \omega & \omega^{7} & 1 & \omega^{5} & \omega^3 & \omega^3 & \omega^{7} & 1 
                \end{array}
                \right)
        \end{align*}
    \end{footnotesize}
    be a generator matrix of a $[14,4,10]_9$ Euclidean self-orthogonal code. 
    
    $\quad$ Taking $l=1$, $\lambda=1$, $\alpha=\omega^{6}$ and $\beta=\omega$, we have $3\alpha^2+\beta^2=\omega^{2}\neq 1$. 
        From Theorem \ref{th.arbitrary hull codess with length n+1}, we can set $G'={\rm diag}(\omega,\ \omega,\ \omega,\ 1)$ and 
        $\mathbf{a}=(\omega^{6}\ \omega^{6}\ \omega^{6}\ 0)^T$, then $G_1$ in Theorem \ref{th.arbitrary hull codess with length n+1} 
        can be written as 
        \begin{footnotesize}
            \begin{align*}
                G_1= \left( 
                    \begin{array}{ccccccccccccccc}
                        \omega & 0 & 0 & 0 & \omega^{3} & \omega^3 & \omega^3 & 1 & 2 & \omega^{5} & 0 & 2 & \omega^{3} & \omega^6  & \omega^6\\
                        0 & \omega & 0 & 0 & \omega^{2} & 0 & \omega^{7} & \omega^{7} & \omega^2 & \omega & 2 & \omega^{5} & \omega^6 & 1  & \omega^6\\
                        0 & 0 & \omega & 0 & 2 & \omega^{3} & \omega^{5} & 2 & 0 & \omega & \omega^{5} & \omega^2 & 2 & \omega  & \omega^6\\
                        0 & 0 & 0 & 1 & \omega^{5} & \omega & \omega & \omega^{7} & 1 & \omega^{5} & \omega^3 & \omega^3 & \omega^{7} & 1 & 0 
                    \end{array}
                    \right).
            \end{align*}
        \end{footnotesize}

        $\quad$ Computed with the Magma software package \cite{BCP1997}, we know that the code $\mathcal{C}_1$ generated by $G_1$ is 
        a $[15,4,10]_{9}$ code with $1$-dimensional Euclidean hull and its Euclidean dual code $\mathcal{C}_1^{\bot_0}$ 
        has parameters $[15,11,4]_{9}$. 
        Note that $\mathcal{C}_1$ is almost-distance-optimal and  $\mathcal{C}_1^{\bot_0}$ is distance-optimal \cite{RefJ-1}. 

        \item [\rm (2)] Let $p=13$, $h=2$, i.e., $q=13^2$ and $\omega$ be a primitive element of $\mathbb{F}_{13^2}$.
        We consider the $[11,4,7]_{13^2}$ Euclidean self-orthogonal code, which constructed by taking $r=0$ in 
        \cite[Example 1]{RefJ-2} and has a generator matrix  
        \begin{footnotesize}
        \begin{align*}
            G= \left( 
                \begin{array}{ccccccccccc}
                    1 & 0 & 0 & 0 & \omega^{133} & 0 & 9 & \omega^{35} & \omega^{119} & 2 & \omega^{105}  \\
                    0 & 1 & 0 & 0 & \omega^{133} & \omega^{21} & 3 & \omega^{63} & \omega^{119} & 5 & 0  \\
                    0 & 0 & 1 & 0 & \omega^{35} & \omega^{35} & 9 & 0 & \omega^{133} & 7 & \omega^{7}  \\
                    0 & 0 & 0 & 1 & \omega^{77} & \omega^{147} & 3 & \omega^{21} & \omega^{105} & 9 & \omega^{133}  
                \end{array}
                \right).
        \end{align*}
    \end{footnotesize}
    
    $\quad$ Taking $l=1$, $\lambda=1$, $\alpha=\omega^{19}$ and $\beta=4$, we have $3\alpha^2+\beta^2=\omega^{134}\neq 1$. 
        From Theorem \ref{th.arbitrary hull codess with length n+1}, we can set $G'={\rm diag}(4,\ 4,\ 4,\ 1)$ and 
        $\mathbf{a}=(\omega^{19}\ \omega^{19}\ \omega^{19}\ 0)^T$, then $G_1$ in Theorem \ref{th.arbitrary hull codess with length n+1} 
        can be written as 
        \begin{footnotesize}
        \begin{align*}
            G_1 = \left( 
                \begin{array}{cccccccccccc}
                    4 & 0 & 0 & 0 & \omega^{133} & 0 & 9 & \omega^{35} & \omega^{119} & 2 & \omega^{105} & \omega^{19}  \\
                    0 & 4 & 0 & 0 & \omega^{133} & \omega^{21} & 3 & \omega^{63} & \omega^{119} & 5 & 0 & \omega^{19} \\
                    0 & 0 & 4 & 0 & \omega^{35} & \omega^{35} & 9 & 0 & \omega^{133} & 7 & \omega^{7} & \omega^{19} \\
                    0 & 0 & 0 & 1 & \omega^{77} & \omega^{147} & 3 & \omega^{21} & \omega^{105} & 9 & \omega^{133} & 0
                \end{array}
                \right).
        \end{align*}
        \end{footnotesize}
       
        $\quad$ Computed with the Magma software package \cite{BCP1997}, we know that the code $\mathcal{C}_1$ generated by $G_1$ is 
        a $[12,4,7]_{13^2}$ code with $1$-dimensional Euclidean hull and its Euclidean dual code $\mathcal{C}_1^{\bot_0}$ 
        has parameters $[12,8,3]_{13^2}$. Since $S(\C_1)=S(\C_1^{\bot_0})=2\leq \{3,6\}$, $\C_1$ is an Euclidean self-orthogonal $2$-MDS code. 
    \end{enumerate}
\end{example}

\section{ Two new families of Hermitian self-orthogonal MDS codes}\label{sec5}

In this section, we focus on the Hermitian inner product, i.e., $q=p^h$, where $h$ is even, and construct two classes of Hermitian self-orthogonal 
MDS codes by employing the criterion displayed in Lemma \ref{lemma.GUO___Hermitian self-orthogonal GRS}.  
It should be emphasized that $p^{\frac{h}{2}}=\sqrt{q}$ is also a prime power in this case. 
For convenience, from now on, we use the notation $\sqrt[]{q}$ instead of $p^\frac{h}{2}$.

\subsection{The first class of Hermitian self-orthogonal MDS codes}\label{sec5.1}
For our first construction, we deal with a multiplicative coset decomposition of $\mathbb{F}_{q}$, which was ever studied in \cite{RefJ-5}. 
Suppose $n'\mid (q-1)$ and denote $n'=\frac{n'}{\gcd(n',\sqrt{q}+1)}\cdot \gcd(n',\sqrt{q}+1)$. 
For simplicity, let $n_1=\frac{n'}{\gcd(n',\sqrt{q}+1)}$ and $n_2=\frac{n'}{n_1}=\gcd(n',\sqrt{q}+1)$. Same as \cite{RefJ-5}, let $G$ and $H$ 
be the subgroups of $\mathbb{F}_{q}^*$ generated by $\omega^{\frac{q-1}{n'}}$ and $\omega^{\frac{\sqrt{q}+1}{n_2}}$, respectively, where $\omega$ 
is a primitive element of $\mathbb{F}_q$. Then $|G|=n'$, $|H|=(\sqrt{q}-1)n_2$ and $G$ is a subgroup of $H$. 
Hence, there exist $\beta_1,\beta_2,\cdots,\beta_{\frac{\sqrt{q}-1}{n_1}}\in H$ such that $\{\beta_bG\}_{b=1}^{\frac{\sqrt{q}-1}{n_1}}$ 
are all cosets of $H/G$. Denote $A_b=\beta_bG$, then it is easy to check that $A_i\cap A_j=\emptyset$ for any $1\leq i\neq j\leq \frac{\sqrt{q}-1}{n_1}$. 

Assume that $n=tn'$ with $1\leq t\leq \frac{\sqrt{q}-1}{n_1}$, and denote 
\begin{equation}\label{eq.ui_decomposution}
    \mathcal{A}=\bigcup_{b=1}^{t} A_b=\{a_1,a_2,\cdots,a_n\}.
\end{equation}
Then, we can get the following result, which is vital for our construction.

\begin{lemma}\label{lem.Con_Hermitian self-orthogonal GRS code_ui}
Let notations be the same as before. Given $1\leq i\leq n$, suppose $a_i\in A_b$ for some $1\leq b\leq t$, then
\begin{equation}\label{eq.Con_Hermitian self-orthogonal GRS code_ui}
    u_i=n'^{-1}a_i\beta_b^{-n'}\prod_{1\leq s\leq t,s\neq b}(\beta_b^{n'}-\beta_s^{n'})^{-1}.
\end{equation}
Moreover, $a_i^{n'-1}u_i\in \mathbb{F}_{\sqrt{q}}^*$.
\end{lemma}
\begin{proof}
Let $\theta=\omega^{\frac{q-1}{n'}}$ be the generator of $G$ and $a_i=\beta_b\theta^\ell$ for some $1\leq \ell\leq n'$. 
By \cite[Lemma 3.7]{RefJ-5}, instantly, we can see that Equation (\ref{eq.Con_Hermitian self-orthogonal GRS code_ui}) holds 
and $\beta_s^{n'} \in \mathbb{F}_{\sqrt{q}}^*$. Moreover, we have  
\[\begin{split}
    a_i^{n'-1}u_i = & a_i^{-1}\cdot a_i^{n'}\cdot n'^{-1}a_i\beta_b^{-n'}\prod_{1\leq s\leq t,s\neq b}(\beta_b^{n'}-\beta_s^{n'})^{-1} \\
    = & (\beta_b\theta^\ell)^{n'} \cdot n'^{-1}\beta_b^{-n'}\prod_{1\leq s\leq t,s\neq b}(\beta_b^{n'}-\beta_s^{n'})^{-1} \\
    = & n'^{-1}\prod_{1\leq s\leq t,s\neq b}(\beta_b^{n'}-\beta_s^{n'})^{-1}.
\end{split} \] 
Hence, $a_i^{n'-1}u_i\in \mathbb{F}_{\sqrt{q}}^*$, which completes the proof.
\end{proof}

\begin{theorem}\label{th.Con_Hermitian self-orthogonal GRS codes}
Let $n=tn'$, where $n'\mid (q-1)$ and $1\leq t\leq \frac{\sqrt{q}-1}{n_1}$ with $n_1=\frac{n'}{\gcd(n',\sqrt{q}+1)}$. 
Then there exists an $[n,k,n-k+1]_{q}$ Hermitian self-orthogonal MDS code, where $1\leq k\leq \lfloor \frac{\sqrt{q}+n'(t-1)}{\sqrt{q}+1} \rfloor$. 
\end{theorem}
\begin{proof}
    Let notations be the same as before. Set $\lambda=1\in \mathbb{F}_{q}^*$ and $h(x)=x^{n'-1}\in \mathbb{F}_q[x]$. From Lemma \ref{lem.Con_Hermitian self-orthogonal GRS code_ui}, 
    we have $\lambda u_ih(a_i)=a_i^{n'-1}u_i\in \mathbb{F}_{\sqrt{q}}^*$. Hence, there is $v_i\in \mathbb{F}_{q}^*$ such that $\lambda u_ih(a_i)=v_i^{\sqrt{q}+1}$ for any 
    $1\leq i\leq n$. Set $\mathbf{v}=(v_1,v_2,\dots,v_n)$.
    
    For $1\leq k\leq \lfloor \frac{\sqrt{q}+n'(t-1)}{\sqrt{q}+1} \rfloor$, consider any codeword $\mathbf{c}=(v_1f(a_1),v_2f(a_2),\dots,v_nf(a_n))\in \GRS_k(\mathbf{a},\mathbf{v})$ 
    with $\deg(f(x))\leq k-1$. Note that 
\[\begin{split}
    \deg(h(x)f^{\sqrt{q}}(x))\leq n'-1+\sqrt{q}(k-1)\leq n-k-1,
\end{split} \]
then by Lemma \ref{lemma.GUO___Hermitian self-orthogonal GRS}, $\mathbf{c}\in \GRS_k(\mathbf{a},\mathbf{v})^{\bot_H}$. 
It follows that $\GRS_k(\mathbf{a},\mathbf{v})\subseteq \GRS_k(\mathbf{a},\mathbf{v})^{\bot_H}$, i.e., $\GRS_k(\mathbf{a},\mathbf{v})$ is 
an $[n,k,n-k+1]_{q}$ Hermitian self-orthogonal MDS code. 
\end{proof}

\begin{remark}\label{rem11.th38.Hermitian self-orthogonal} $\quad$ 
    \begin{enumerate}
        \item [\rm (1)] For the same code length $n=tn'$, the authors proved that for each 
        $1\leq k\leq \lfloor \frac{\sqrt{q}+n}{\sqrt{q}+1} \rfloor$, there exists an $[n,k,n-k+1]_q$ MDS code $\mathcal{C}$ with 
        $\dim(\Hull_{\frac{h}{2}}(\mathcal{C}))=l$ in \cite[Theorem 3.8]{RefJ-5}, where $0\leq l\leq k-1$. 
        Note that the Herimitian self-orthogonality means $\dim(\Hull_{\frac{h}{2}}(\C))=k$. 
        Therefore, Theorem \ref{th.Con_Hermitian self-orthogonal GRS codes} implies that there also 
        exists an $[n,k,n-k+1]_q$ MDS code with $l=k$ for $1\leq k\leq \lfloor \frac{\sqrt{q}+n'(t-1)}{\sqrt{q}+1} \rfloor$.

        \item [\rm (2)] In Table \ref{tab:3}, we list some Hermitian self-orthogonal MDS codes over $\mathbb{F}_{3^4}$ from 
        Theorem \ref{th.Con_Hermitian self-orthogonal GRS codes}. 
    \end{enumerate}
\end{remark}

\begin{table}[H]
    \centering
    \caption{Some Hermitian self-orthogonal MDS codes from Theorem \ref{th.Con_Hermitian self-orthogonal GRS codes} for $q=3^4$}
    \label{tab:3}       
        \begin{tabular}{cccc||cccc}
            \hline
           $n'$ & $t$ & $[n,k,d]_{q}$ & $k$ & $n'$ & $t$ & $[n,k,d]_{q}$ & $k$\\
            \hline
            2 & 7 & $[14,k,15-k]_{3^4}$ & $1\leq k\leq 2$ & 2 & 8 & $[16,k,17-k]_{3^4}$ & $1\leq k\leq 2$ \\
            5 & 4 & $[20,k,21-k]_{3^4}$ & $1\leq k\leq 2$ & 5 & 5 & $[25,k,26-k]_{3^4}$ & $1\leq k\leq 2$ \\
            5 & 6 & $[30,k,31-k]_{3^4}$ & $1\leq k\leq 3$ & 5 & 7 & $[35,k,36-k]_{3^4}$ & $1\leq k\leq 3$ \\
            5 & 8 & $[40,k,41-k]_{3^4}$ & $1\leq k\leq 4$ & 10 & 5 & $[50,k,51-k]_{3^4}$ & $1\leq k\leq 4$ \\
            10 & 6 & $[60,k,61-k]_{3^4}$ & $1\leq k\leq 5$ & 10 & 7 & $[70,k,71-k]_{3^4}$ & $1\leq k\leq 6$ \\
            10 & 8 & $[80,k,81-k]_{3^4}$ & $1\leq k\leq 7$ &  &  & &  \\
            \hline
        \end{tabular}
    \end{table}

    Recently, Guo et al. \cite{RefJ-4} constructed two classes of Hermitian self-orthogonal MDS codes listed in Table \ref{tab:4} with similar parameters 
    to ours. We do some detailed comparisons in the following remark and illustrate that our codes are new and better. 

    \begin{table}[H]
        \caption{Two classes of Hermitian self-orthogonal MDS codes with similar parameters to ours}
        \label{tab:4}       
        \centering
            \begin{tabular}{cllc}
                \hline
               Class & Code Length & Dimension & Ref.\\
                \hline
                1 & $n=m(\sqrt{q}-1)$, $1\leq m\leq \sqrt{q}$ & $1\leq k\leq \lfloor \frac{m\sqrt{q}-1}{\sqrt{q}+1} \rfloor$  & \cite{RefJ-4} \\
                2 & $n=s(\sqrt{q}+1)$, $1\leq s\leq \sqrt{q}-1$ & $1\leq k\leq s-1$  & \cite{RefJ-4} \\
                \hline
            \end{tabular}
    \end{table}

\begin{remark}\label{rem12.th38.comparisons.Hermitian self-orthogonal} $\quad$ 
    \begin{enumerate}
        \item [\rm (1)] For Class $1$ in Table \ref{tab:4}, taking $n'=\sqrt{q}-1$ in Theorem \ref{th.Con_Hermitian self-orthogonal GRS codes}, 
        we have $n_1=\frac{\sqrt{q}-1}{\gcd(\sqrt{q}-1,\sqrt{q}+1)}=\sqrt{q}-1$ if $p$ is even and $\frac{\sqrt{q}-1}{2}$ if $p$ is 
        odd from Lemma \ref{lemma.gcd}. In these two cases, $t=1$ and $1\leq t\leq 2$, respectively. Then the same form of length $n=t(\sqrt{q}-1)$ 
        can be obtained in Theorem \ref{th.Con_Hermitian self-orthogonal GRS codes}. Due to the other possible values of $n'$, it is clear that Theorem 
        \ref{th.Con_Hermitian self-orthogonal GRS codes} can produce many codes with other forms of lengths for different $q$. Some explicit examples are 
        given in the consecutive discussions.

        \item [\rm (2)] Looking closely at Table \ref{tab:3}, we can deduce that the construction in Theorem \ref{th.Con_Hermitian self-orthogonal GRS codes} can 
        contain all the results from Class $2$ in Table \ref{tab:4}, and can also generate additional codes such as $[14,k_1,15-k_1]_{81}$, $[16,k_2,17-k_2]_{81}$, 
        $[25,k_3,26-k_3]_{81}$ and $[35,k_4,36-k_4]_{81}$. 

        \item [\rm (3)] The reason for the phenomenon in {\rm (2)} above is uncomplicated. We only need to take $n'=\sqrt{q}+1$ in 
        Theorem \ref{th.Con_Hermitian self-orthogonal GRS codes}, then $n_1=\frac{\sqrt{q}+1}{\gcd(\sqrt{q}+1,\sqrt{q}+1)}=1$ 
        and $1\leq t\leq \sqrt{q}-1$. Therefore, we can get $[t(\sqrt{q}+1),k,t(\sqrt{q}+1)-k+1]_{q}$ Hermitian self-orthogonal MDS codes 
        from Theorem \ref{th.Con_Hermitian self-orthogonal GRS codes}, where $1\leq t\leq \sqrt{q}-1$ and $1\leq k\leq t-1$.  
        Obviously, these codes are all the results given by Class $2$ in Table \ref{tab:4}.

        \item [\rm (4)] Also, on the basis of the discussion of $(2)$ above, considering some same lengths, the range of the dimensions of our codes 
        are more flexible in some cases. For example, Table \ref{tab:3} gives $[30,k,31-k]_{81}$ Hermitian self-orthogonal codes, where $1\leq k\leq 3$. 
        However, if we construct $[30,k,31-k]_{81}$ codes from Class $2$ of Table \ref{tab:4}, we have $s=3$, and thus, $1\leq k\leq 3-1=2$ is required. 

        \item [\rm (5)] In terms of more length forms and more flexible dimensions, we can confirm that our codes are better. 
        And in conjunction with the comparisons in \cite{RefJ-4}, these codes are also new. 
    \end{enumerate}
\end{remark}

\subsection{The second class of Hermitian self-orthogonal MDS codes}\label{sec5.2}

In this subsection, we only consider the case odd $q$. Then, the second construction comes from the direct product of two cyclic subgroups of $\mathbb{F}_{q}$, 
which was recently studied in \cite{RefJ-12}. Let $\omega$ be a primitive element of $\mathbb{F}_{q}$. 
Set $\beta_1=\omega^{x_1}$ and $\beta_2=\omega^{x_2}$, where $x_1,\ x_2$ are two integers. According to \cite{RefJ-12}, we know that the direct product 
$\langle \beta_1 \rangle \bigotimes \langle \beta_2\rangle $ is a subgroup of $\mathbb{F}_q^*$ with order ${\rm ord}(\beta_1)\cdot {\rm ord}(\beta_2)$. 
Furthermore, we have ${\rm ord}(\beta_1)=\frac{q-1}{\gcd(q-1,x_1)}$ and ${\rm ord}(\beta_1)=\frac{q-1}{\gcd(q-1,x_2)}$ here.

Assume that $n=n_1n_2$ with $1\leq n_1\leq {\rm ord}(\beta_1)$ and $n_2={\rm ord}(\beta_2)$, and denote 
\begin{align}
    \mathcal{N}=\bigcup_{i=1}^{n_1} N_i=\{a_1,a_2,\cdots,a_n\},
\end{align}
where $N_i=\{\beta_1^i \beta_2^j:\ 1\leq j\leq n_2\}$ for $1\leq i\leq n_1$. 
Then, it is easy to check that $N_i\cap N_j=\emptyset$ for any $1\leq i\neq j\leq n_1$ and 
we can derive the following result, which is vital for our construction.

\begin{lemma}\label{lem.Con_Hermitian self-orthogonal GRS code222_ui}
    Let notations be the same as before. Given $1\leq i\leq n$, suppose $a_i\in N_s$ for some $1\leq s\leq n_1$. 
    If $(q-1)\mid {\rm lcm}(x_1,x_2)$ and $\gcd(x_2,q-1)\mid x_1(\sqrt{q}-1)$, then
    \begin{equation}\label{eq.Con_Hermitian self-orthogonal GRS code222_ui}
        u_i=a_i\beta_1^{-sn_2}n_2^{-1}\prod_{1\leq s'\leq n_1,s'\neq s}(\beta_1^{sn_2}-\beta_1^{s'n_2})^{-1}.
    \end{equation}
    Moreover, $a_i^{n_2-1}u_i\in \mathbb{F}_{\sqrt{q}}^*$.
\end{lemma}
\begin{proof}
    Since $a_i\in N_s=\{\beta_1^s \beta_2^j:\ 1\leq j\leq n_2\}$ for some $1\leq s\leq n_1$, 
    we can set $a_i=\beta_1^s\beta_2^t$ for some $1\leq t\leq n_2$. 
    By \cite[Lemma 3.1]{RefJ-12}, instantly, we can see that 
    Equation (\ref{eq.Con_Hermitian self-orthogonal GRS code222_ui}) holds and $\beta_1^{n_2} \in \mathbb{F}_{\sqrt{q}}^*$. 
    Moreover, we have 
    \[\begin{split}
        a_i^{n_2-1}u_i = & a_i^{-1}\cdot a_i^{n_2}\cdot a_i\beta_1^{-sn_2}n_2^{-1}\prod_{1\leq s'\leq n_1,s'\neq s}(\beta_1^{sn_2}-\beta_1^{s'n_2})^{-1} \\
                       = & \beta_1^{sn_2}\beta_2^{tn_2}\cdot \beta_1^{-sn_2}n_2^{-1}\prod_{1\leq s'\leq n_1,s'\neq s}(\beta_1^{sn_2}-\beta_1^{s'n_2})^{-1} \\
                       = & n_2^{-1}\prod_{1\leq s'\leq n_1,s'\neq s}(\beta_1^{sn_2}-\beta_1^{s'n_2})^{-1}.
    \end{split} \] 
    The last equation holds for the facts ${\rm ord}(\beta_2)=n_2$ and $\beta_2^{tn_2}=1$. 
    Hence, $a_i^{n_2-1}u_i\in \mathbb{F}_{\sqrt{q}}^*$, which completes the proof.
\end{proof}

\begin{theorem}\label{th.Con_Hermitian self-orthogonal GRS codes222}
    Let $q=p^h$ be odd. Assume that $(q-1)\mid {\rm lcm}(x_1,x_2)$ and $\gcd(x_2,q-1)\mid x_1(\sqrt{q}-1)$, where $x_1$ and $x_2$ are two integers. 
    Let $n=n_1n_2$ with $1\leq n_1\leq \frac{q-1}{\gcd(q-1,x_1)}$ and $n_2=\frac{q-1}{\gcd(q-1,x_2)}$. 
    Then there exists an $[n,k,n-k+1]_{q}$ Hermitian self-orthogonal MDS code, 
    where $1\leq k\leq \lfloor \frac{\sqrt{q}+n_2(n_1-1)}{\sqrt{q}+1} \rfloor$. 
    \end{theorem}
    \begin{proof}
        Let notations be the same as before. Set $\lambda=1\in \mathbb{F}_{q}^*$ and $h(x)=x^{n_2-1}\in \mathbb{F}_q[x]$. 
        From Lemma \ref{lem.Con_Hermitian self-orthogonal GRS code222_ui}, 
        we have $\lambda u_ih(a_i)=a_i^{n_2-1}u_i\in \mathbb{F}_{\sqrt{q}}^*$. 
        Hence, there is $v_i\in \mathbb{F}_{q}^*$ such that $\lambda u_ih(a_i)=v_i^{\sqrt{q}+1}$ for any 
        $1\leq i\leq n$. Set $\mathbf{v}=(v_1,v_2,\dots,v_n)$.
        
        For $1\leq k\leq \lfloor \frac{\sqrt{q}+n_2(n_1-1)}{\sqrt{q}+1} \rfloor$, 
        consider any codeword $\mathbf{c}=(v_1f(a_1),v_2f(a_2),\dots,v_nf(a_n))\in \GRS_k(\mathbf{a},\mathbf{v})$ 
        with $\deg(f(x))\leq k-1$. Note that 
    \[\begin{split}
        \deg(h(x)f^{\sqrt{q}}(x))\leq n_2-1+\sqrt{q}(k-1)\leq n-k-1,
    \end{split} \]
    then by Lemma \ref{lemma.GUO___Hermitian self-orthogonal GRS}, $\mathbf{c}\in \GRS_k(\mathbf{a},\mathbf{v})^{\bot_H}$. 
    It follows that $\GRS_k(\mathbf{a},\mathbf{v})\subseteq \GRS_k(\mathbf{a},\mathbf{v})^{\bot_H}$, i.e., $\GRS_k(\mathbf{a},\mathbf{v})$ is 
    an $[n,k,n-k+1]_{q}$ Hermitian self-orthogonal MDS code. 
    \end{proof}

    \begin{remark}\label{rem13.th40.Hermitian self-orthogonal222} $\quad$ 

            {\rm (1)} For the same code length $n=n_1n_2$, 
            the authors proved that for each $1\leq k\leq \lfloor \frac{\sqrt{q}+n}{\sqrt{q}+1} \rfloor$, 
            there exists an $[n,k,n-k+1]_q$ MDS code $\mathcal{C}$ with $\dim(\Hull_{\frac{h}{2}}(\mathcal{C}))=l$ in \cite[Theorem 3.1]{RefJ-12}, 
            where $0\leq l\leq k-1$. Similar to Remark \ref{rem11.th38.Hermitian self-orthogonal}, Theorem \ref{th.Con_Hermitian self-orthogonal GRS codes222} 
            implies that there also exists an $[n,k,n-k+1]_q$ MDS code with $l=k$ 
            for $1\leq k\leq \lfloor \frac{\sqrt{q}+n_2(n_1-1)}{\sqrt{q}+1} \rfloor$. 

            {\rm (2)} Except for two known classes of Hermitian self-orthogonal MDS codes listed in Table \ref{tab:4}, 
            as far as we know, there are many other known constructions of Hermitian self-orthogonal MDS codes 
            in the literature. In Table \ref{tab:5}, we list some of these codes 
            with similar parameters to Theorem \ref{th.Con_Hermitian self-orthogonal GRS codes222}.  
            One can easily see that our codes are new in view of the different conditions and length forms. 
            And the same judgment can also be derived from the detailed discussions in \cite{RefJ-12}. 
            
            {\rm (3)} In Table \ref{tab:6}, some Hermitian self-orthogonal MDS codes over $\mathbb{F}_{13^2}^*$, $\mathbb{F}_{17^2}^*$ 
            and $\mathbb{F}_{5^4}^*$ from Theorem \ref{th.Con_Hermitian self-orthogonal GRS codes222} are listed. 
    \end{remark}

        \begin{table}[H]
            \caption{Some other known Hermitian self-orthogonal MDS codes with similar parameters to ours}
            \label{tab:5}       
            \centering
                \begin{tabular}{cllc}
                    \hline
                   Class & Code Length & Dimension & Ref.\\
                    \hline
                    1 & $n=q+1$ & $0\leq k\leq \sqrt{q}+1$ & \cite{RefJ-7,RefJ-8} \\
                    2 & $n=q-l$, $0\leq l\leq \sqrt{q}-2$ & $k\leq \sqrt{q}-l-1$ & \cite{RefJ-7} \\
                    3 & $n=t\sqrt{q}$, $1\leq t\leq \sqrt{q}$ & $1\leq k\leq \lfloor \frac{t\sqrt{q}+\sqrt{q}-1}{\sqrt{q}+1} \rfloor$ & \cite{RefJ-6} \\
                    4 & $n=bm(\sqrt{q}+1)$, $1\leq b\leq 2a$, $\sqrt{q}=2am+1$ & $1\leq k\leq (a+1)m$ & \cite{RefJ-14} \\
                    5 & $n=bm(\sqrt{q}-1)$, $1\leq b\leq 2a$, $\sqrt{q}=2am-1$ & $1\leq k\leq (a+1)m-2$ & \cite{RefJ-14} \\
                    \hline
                \end{tabular}
        \end{table}
        
        \begin{table}[H]
            \centering
            \caption{Some Hermitian self-orthogonal MDS codes from Theorem \ref{th.Con_Hermitian self-orthogonal GRS codes222} for $q=13^2$, $17^2$ and $5^4$}
            \label{tab:6}       
                \begin{tabular}{ccc||ccc}
                    \hline
                     $(q,x_1,x_2,n_1,n_2)$ & $[n,k,d]_{q}$ & $k$ & $(q,x_1,x_2,n_1,n_2)$ & $[n,k,d]_{q}$ & $k$\\
                    \hline
                    $(13^2,14,24,5,7)$ & $[35,k,36-k]_{13^2}$ & $1\leq k\leq 2$ & $(13^2,14,24,7,7)$ & $[49,k,50-k]_{13^2}$ & $1\leq k\leq 3$ \\
                    $(13^2,14,24,9,7)$ & $[63,k,64-k]_{13^2}$ & $1\leq k\leq 4$ & $(13^2,14,24,11,7)$ & $[77,k,78-k]_{13^2}$ & $1\leq k\leq 5$ \\

                    $(17^2,18,32,9,9)$ & $[81,k,82-k]_{17^2}$ & $1\leq k\leq 4$ & $(17^2,18,32,11,9)$ & $[99,k,100-k]_{17^2}$ & $1\leq k\leq 5$ \\
                    $(17^2,18,32,13,9)$ & $[117,k,118-k]_{17^2}$ & $1\leq k\leq 6$ & $(17^2,18,32,15,9)$ & $[135,k,136-k]_{17^2}$ & $1\leq k\leq 7$ \\

                    $(5^4,26,48,7,13)$ & $[91,k,92-k]_{5^4}$ & $1\leq k\leq 3$ & $(5^4,26,48,11,13)$ & $[143,k,144-k]_{5^4}$ & $1\leq k\leq 5$ \\
                    $(5^4,26,48,15,13)$ & $[195,k,196-k]_{5^4}$ & $1\leq k\leq 7$ & $(5^4,26,48,19,13)$ & $[247,k,248-k]_{5^4}$ & $1\leq k\leq 9$ \\
                    $(5^4,26,48,21,13)$ & $[273,k,274-k]_{5^4}$ & $1\leq k\leq 10$ & $(5^4,26,48,23,13)$ & $[299,k,300-k]_{5^4}$ & $1\leq k\leq 11$ \\

                    \hline
                \end{tabular}
            \end{table}

\section{New EAQECCs and MDS EAQECCs}\label{sec6}

In this section, we apply all our results to the construction of new $q$-ary and $\sqrt{q}$-ary EAQECCs as well as MDS EAQECCs. 
The lengths of many $\sqrt{q}$-ary MDS EAQECCs are greater than $\sqrt{q}+1$ and their minimum distances are 
greater than $\lceil \frac{\sqrt{q}}{2} \rceil$. 
\subsection{EAQECCs derived from classical codes}\label{sec6.1}

Similar to classical codes, there is a so-called entanglement-assisted quantum Singleton bound (EA-quantum Singleton bound). 

\begin{lemma}\label{lem2.2}{\rm (EA-quantum Singleton Bound \cite{RefJ50})}
Let $\mathcal{Q}$ be an $[[n,k,d;c]]_q$ entanglement-assisted quantum error-correcting code. If $2d\leq n+2$, then
\begin{align}
    k\leq n+c-2(d-1).
\end{align}
\end{lemma}

\begin{remark}\label{rem14.EA-quantum Singleton bound} $\quad$
\begin{enumerate}
    \item [\rm (1)] An EAQECC for which equality holds in the EA-quantum Singleton bound is called an MDS EAQECC. 
    \item [\rm (2)] It is easy to verify that if a classical code $\C$ is an MDS code and $2d\leq n+2$, then the EAQECC constructed by it is an MDS EAQECC.

\end{enumerate}
\end{remark} 

As mentioned before, in an EAQECC, the parameter $c$ can be calculated with $\dim(\Hull(\C))$. 
As a unified description, Cao \cite{RefJ3} proposed a method of constructing EAQECCs from  
codes with prescribed dimensional Galois hull and we rephrase it here.

\begin{lemma}{\rm(\cite{RefJ3})}\label{lemma_Galois EAQECC}
    Let $\mathcal{C}$ be an $[n,k,d]_q$ code. Then there exists an $[[n,k-\dim(\Hull_e(\mathcal{C})),d;n-k-\dim(\Hull_e(\mathcal{C}))]]_q$ EAQECC $\mathcal{Q}$.
\end{lemma}

Furthermore, in Corollary \ref{coro.C and C dual Hull equals}, we prove that $\dim(\Hull_e(\mathcal{C}))=\dim(\Hull_e(\mathcal{C}^{\bot_e}))$ 
for each $0\leq e\leq h-1$. Therefore, we can actually construct EAQECCs from both $\C$ and $\C^{\bot_e}$ by Lemma \ref{lemma_Galois EAQECC}. 
We give the explicit constructions as follows.  

\begin{corollary}\label{coro_Galois EAQECC}
    Let $\mathcal{C}$ be an $[n,k,d]_q$ code and $\mathcal{C}^{\bot_e}$ be its $e$-Galois dual code with 
    parameters $[n,n-k,d^{\bot_e}]_q$. Then there exists an $[[n,k-\dim(\Hull_e(\mathcal{C})),d;n-k-\dim(\Hull_e(\mathcal{C}))]]_q$ EAQECC 
    $\mathcal{Q}$ and an $[[n,n-k-\dim(\Hull_e(\mathcal{C})),d^{\bot_e};k-\dim(\Hull_e(\mathcal{C}))]]_q$ EAQECC $\mathcal{Q}'$. 
\end{corollary}

Notably, in \cite{RefJ12}, the authors showed that one can construct $\sqrt[]{q}$-ary EAQECCs 
from $q$-ary codes when the Hermitian inner product is taken into account.  

\begin{lemma}{\rm (\cite{RefJ12})}\label{lem.Hermitian EAEECCs}
    Let $\mathcal{C}$ be an $[n,k,d]_q$ code and $\mathcal{C}^{\bot_\frac{h}{2}}$ be its Hermitian dual code with parameters $[n,n-k,d^{\bot_\frac{h}{2}}]_q$. 
    Then there exists 
    an $[[n,k-\dim(\Hull_{\frac{h}{2}}(\mathcal{C})),d;n-k-\dim(\Hull_{\frac{h}{2}}(\mathcal{C}))]]_{\sqrt{q}}$ EAQECC $\mathcal{Q}$ and 
    an $[[n,n-k-\dim(\Hull_{\frac{h}{2}}(\mathcal{C})),d^{\bot_{\frac{h}{2}}};k-\dim(\Hull_{\frac{h}{2}}(\mathcal{C}))]]_{\sqrt{q}}$ EAQECC $\mathcal{Q}'$. 
\end{lemma}

From Corollaries \ref{coro.new_arbitrary dimension e-Galois hull codes with length n+2i and n+2i+1}-\ref{coro.new_arbitrary dimension Hermitian hull codes with length n+2i and n+2i+1} and 
Theorem \ref{th.arbitrary hull codess with length n+1}, 
we can construct EAQECCs of larger length using Corollary \ref{coro_Galois EAQECC} and Lemma \ref{lem.Hermitian EAEECCs}. 
Let notations be the same as before. In the sequel, we further denote the minimum distance of the $e$-Galois dual code of 
$\C_{i}$ ($i\geq 0$ and $i\neq 1$), $\C_{2i}$ ($i\geq 0$) and $\C_1$ by $\widetilde{d}_{i}^{\bot_e}$, $\widetilde{d}_{2i}^{\bot_e}$ 
and $\widetilde{d}_{1}^{\bot_e}$, respectively.
Therefore, the following results are straightforward. 

\begin{corollary}\label{coro.EAQECCS for Galois hulls}
    Let $q\geq 3$ and $\mathcal{C}$ be an $[n,k,d]_q$ $e$-Galois self-orthogonal code.  Then the following statements hold. 
    \begin{itemize}
        \item [\rm (1)] If Equation (\ref{eq.th.n+i}) holds, then there exists 
        an $[[n+i,k-l,\widetilde{d}_{i};n+i-k-l]]_q$ EAQECC $\mathcal{Q}_{i}$ and 
        an $[[n+i,n+i-k-l,\widetilde{d}_{i}^{\bot_e};k-l]]_q$ EAQECC $\mathcal{Q}_{i}'$ 
        for $i\geq 0$, $i\neq 1$ and $0\leq l\leq k$, where $d\leq \widetilde{d}_{i}\leq n+i+1-k$. 
        
        \item [\rm (2)] There exists an $[[n+1,1,\widetilde{d}_{1};n-2k+2]]_q$ EAQECC $\mathcal{Q}_{1}$ and 
        an $[[n+1,n-2k+2,\widetilde{d}_{1}^{\bot_e};1]]_q$ EAQECC $\mathcal{Q}_{1}'$, where $d\leq \widetilde{d}_1\leq n+2-k$.
        
        \item [\rm (3)] There exists an $[[n+1,k-l,\widetilde{d}_{1};n+1-k-l]]_q$ EAQECC $\mathcal{Q}_{1}$ and 
        an $[[n+1,n+1-k-l,\widetilde{d}_{1}^{\bot_e};k-l]]_q$ EAQECC $\mathcal{Q}_{1}'$ 
        if and only if $(k-l)\alpha^{p^e+1}+\beta^{p^e+1}\neq 1$ and $\beta^{p^e+1}\neq 1$ hold 
        for some $\alpha,\beta\in \mathbb{F}_q^*$ and prescribed $0\leq l\leq k-2$, where $d\leq \widetilde{d}_1\leq n+2-k$.
    \end{itemize}
\end{corollary}

\begin{corollary}\label{coro.EAQECCS for Euclidean hulls}
    Let $q\geq 3$ and $\mathcal{C}$ be an $[n,k,d]_q$ Euclidean self-orthogonal code. Then the following statements hold. 
    \begin{itemize}
        \item [\rm (1)] If $p$ is even and $h\geq 2$, or $p$ is odd and $h$ is even, there exists an $[[n+i,k-l,\widetilde{d}_{i};n+i-k-l]]_q$ EAQECC $\mathcal{Q}_{i}$ 
        and an $[[n+i,n+i-k-l,\widetilde{d}_{i}^{\bot_0};k-l]]_q$ EAQECC $\mathcal{Q}_{i}'$ 
        for $i\geq 0$, $i\neq 1$ and $0\leq l\leq k$, where $d\leq \widetilde{d}_{i}\leq n+i+1-k$.

        \item [\rm (2)] There exists an $[[n+1,1,\widetilde{d}_{1};n-2k+2]]_q$ EAQECC $\mathcal{Q}_{1}$ and 
        an $[[n+1,n-2k+2,\widetilde{d}_{1}^{\bot_0};1]]_q$ EAQECC $\mathcal{Q}_{1}'$, where $d\leq \widetilde{d}_1\leq n+2-k$.

        \item [\rm (3)] There exists an $[[n+1,k-l,\widetilde{d}_{1};n+1-k-l]]_q$ EAQECC $\mathcal{Q}_{1}$ and 
        an $[[n+1,n+1-k-l,\widetilde{d}_{1}^{\bot_0};k-l]]_q$ EAQECC $\mathcal{Q}_{1}'$ if and only if 
        $(k-l)\alpha^{2}+\beta^{2}\neq 1$ and $\beta^{2}\neq 1$ hold for 
        some $\alpha,\beta\in \mathbb{F}_q^*$ and prescribed $0\leq l\leq k-2$, 
        where $d\leq \widetilde{d}_1\leq n+2-k$.

    \end{itemize}
\end{corollary}

\begin{corollary}\label{coro.EAQECCS for Hermitian hulls}
    Let $\mathcal{C}$ be an $[n,k,d]_q$ Hermitian self-orthogonal code. Then the following statements hold. 
    \begin{itemize}
        \item [\rm (1)] If $q> 4$, there exists 
        an $[[n+i,k-l,\widetilde{d}_{i};n+i-k-l]]_{\sqrt{q}}$ EAQECC $\mathcal{Q}_{i}$ and 
        an $[[n+i,n+i-k-l,\widetilde{d}_{i}^{\bot_\frac{h}{2}};k-l]]_{\sqrt{q}}$ EAQECC $\mathcal{Q}_{i}'$ 
        for $i\geq 0$, $i\neq 1$ and $0\leq l\leq k$, where $d\leq \widetilde{d}_{i}\leq n+i+1-k$.

        \item [\rm (2)] If $q=4$, there exists 
        an $[[n+2i,k-l,\widetilde{d}_{2i};n+2i-k-l]]_{2}$ EAQECC $\mathcal{Q}_{2i}$ and 
        an $[[n+2i,n+2i-k-l,\widetilde{d}_{2i}^{\bot_\frac{h}{2}};k-l]]_{2}$ EAQECC $\mathcal{Q}_{2i}'$ 
        for $i\geq 0$ and $0\leq l\leq k$, where $d\leq \widetilde{d}_{2i}\leq n+2i+1-k$.

        \item [\rm (3)] There exists 
        an $[[n+1,1,\widetilde{d}_{1};n-2k+2]]_{\sqrt{q}}$ EAQECC $\mathcal{Q}_{1}$ and 
        an $[[n+1,n-2k+2,\widetilde{d}_{1}^{\bot_\frac{h}{2}};1]]_{\sqrt{q}}$ EAQECC $\mathcal{Q}_{1}'$, 
        where $d\leq \widetilde{d}_1\leq n+2-k$.

        \item [\rm (4)] There exists 
        an $[[n+1,k-l,\widetilde{d}_{1};n+1-k-l]]_{\sqrt{q}}$ EAQECC $\mathcal{Q}_{1}$ and 
        an $[[n+1,n+1-k-l,\widetilde{d}_{1}^{\bot_\frac{h}{2}};k-l]]_{\sqrt{q}}$ EAQECC $\mathcal{Q}_{1}'$ 
        if and only if $(k-l)\alpha^{\sqrt{q}+1}+\beta^{\sqrt{q}+1}\neq 1$ and 
        $\beta^{\sqrt{q}+1}\neq 1$ hold for some $\alpha,\beta\in \mathbb{F}_q^*$ and prescribed $0\leq l\leq k-2$, 
        where $d\leq \widetilde{d}_1\leq n+2-k$.
    \end{itemize}
\end{corollary}

\subsection{Comparisons with other propagation rules}
In some other articles, the results presented in Corollaries \ref{coro.EAQECCS for Galois hulls}-\ref{coro.EAQECCS for Hermitian hulls} are usually called 
propagation rules.  
Taking the Hermitage case, i.e., Corollary \ref{coro.EAQECCS for Hermitian hulls}, as an example, 
we can express our conclusions that if there exists an $[[n,0,d;n-2k]]_{\sqrt[]{q}}$ EAQECC constructed 
from an $[n,k,d]_q$ Hermitian self-orthogonal code, then all $\sqrt{q}$-ary EAQECCs in Corollary 
\ref{coro.EAQECCS for Hermitian hulls} (1)-(4) exist, i.e., 
\begin{align}\label{eq.propgation rule form}
    [[n,0,d;n-2k]]_{\sqrt{q}}\ \Rightarrow\ \sqrt{q}\textrm{-ary EAQECCs in Corollary \ref{coro.EAQECCS for Hermitian hulls} (1)-(4)}.
\end{align} 
Note that Equation (\ref{eq.propgation rule form}) is a general representation of a propagation rule for EAQECCs.

As revealed in \cite{RefJ B4}, most known propagation rules were designed for stabilizer QECCs. 
In view of the fact that EAQECCs become stabilizer QECCs if $c=0$, we can know that our propagation 
rules are applicable to both stabilizer QECCs and EAQECCs from Equation (\ref{eq.propgation rule form}).  

Note that the core idea of our propagation rules is to increase the length of a known EAQECC, which 
can be accompanied by varying degrees of increase in the minimum distance. 
For ease of reference, we only consider the Hermitian case in this subsection and list some known 
propagation rules for increasing the length of a known EAQECC in Table \ref{tab:propagation}.

\begin{table}[H]
    \centering
    \caption{Propagation rules for increasing the length of a known EAQECC}
    \label{tab:propagation}       
        \begin{tabular}{c|c|c|c}
            \hline
             Class & propagation rule & Limitation & Ref. \\ \hline
             1 & \tabincell{c}{$[[n,k-l,d;n-k-l]]_{\sqrt{q}}$\\ $\Rightarrow \ [[n+1,k-l,d;n-k-l]]_{\sqrt{q}}$} & $-$ & Trivial \\ \hline

             2 & \tabincell{c}{$[[n,k-l,d;n-k-l]]_{\sqrt{q}}$\\ $\Rightarrow\ [[n+1,k-l-1,\widetilde{d}_1;n-k-l]]_{\sqrt{q}}$} & 
             \tabincell{c}{ $k\geq l+1$, \\ $n\geq k+l+1$, $d\leq \widetilde{d}_1\leq d+1$} & \cite[Theorem 16]{RefJ B4} \\ \hline

             3 & \tabincell{c}{$[[n,k-l,d;n-k-l]]_{\sqrt{q}}$\\ $\Rightarrow\ [[n+1,k-l,\widetilde{d}_1;n-k-l-1]]_{\sqrt{q}}$} & 
             \tabincell{c}{ $n\geq k+l+1$, $\widetilde{d}_1\leq d$} & \cite[Theorem 18]{RefJ B4} \\ \hline

             4 & \tabincell{c}{$[[n,k-l,d;n-k-l]]_{\sqrt{q}}$\\ $\Rightarrow\ [[n+r,k-l,\widetilde{d}_r;n-k-l+r]]_{\sqrt{q}}$} & 
             \tabincell{c}{$q\geq 3$ is odd,\\ $0\leq r\leq k-l$, $d\leq \widetilde{d}_r\leq d+r$} & \cite[Proposition 5.4]{RefJ B5} \\ \hline
        \end{tabular}
\end{table}

Detailed comparisons are included in the following remark. 

\begin{remark}\label{rem16.propagation comparisons} $\quad$
    \begin{enumerate}
        \item [\rm (1)] Compared to Classes $1$ and $3$ in Table \ref{tab:propagation}, our propagation rule in 
        Equation (\ref{eq.propgation rule form}) can increase the minimum distance by different levels, while the 
        minimum distance in Class $1$ remains the same, and in Class $3$ may decrease. 
        Additionally, $n-k-l=0$ (i.e., $n=k+l$) is inapplicable for Class $3$, whereas this case is included 
        in Equation (\ref{eq.propgation rule form}). 

        \item [\rm (2)] Compared to Class $2$ in Table \ref{tab:propagation}, our propagation rule in 
        Equation (\ref{eq.propgation rule form}) always keep the dimension instead of reducing it. 
        Additionally, $k-l=0$ (i.e., $k=l$) and $n-k-l=0$ (i.e., $n=k+l$) are inapplicable for Class $2$. 
        However, these cases are included in Equation (\ref{eq.propgation rule form}).

        \item [\rm (3)] Compared to Class $4$ in Table \ref{tab:propagation}, our propagation rule in 
        Equation (\ref{eq.propgation rule form}) has no limitations on $r$. And the limitation odd $q\geq 3$ can 
        be relaxed in our constructions. Particularly, we observe that when $k=l$, Class $4$ is no 
        longer able to generate new EAQECCs, while our propagation rule is still feasible. 

    \end{enumerate}
\end{remark}

\subsection{The performance of derived EAQECCs}

We now discuss the performance of new EAQECCs by rates $\rho$ and net rates $\bar{\rho}$. 
Taking $\mathcal{Q}_{i}'$ in Corollary \ref{coro.EAQECCS for Hermitian hulls} (1) as an example, 
it is easy to check that the net rate $\bar{\rho}$ of $\mathcal{Q}_{i}'$ is 
\begin{align}\label{eq.net rate}
    \bar{\rho}=\frac{(n+i-k-l)-(k-l)}{n+i}=\frac{n+i-2k}{n+i}\geq 0
\end{align}
for any $i\geq 0$ and $i\neq 1$ since the initial $[n,k,d]_q$ code $\C$ is a Hermitian self-orthogonal code, i.e., $2k\leq n$. 
Note that $\bar{\rho}=0$ in Equation (\ref{eq.net rate}) if and only if $n=2k$ and $i=0$. Particularly, $n=2k$ implies that the initial 
Hermitian self-orthogonal code $\C$ is actually Hermitian self-dual. Therefore, in general, we have $\bar{\rho}>0$. 
In addition, if $0\leq \dim(\Hull_H(\C_{i}))=l\leq \min\{\lfloor \frac{n+i}{2} \rfloor-k,k\}$, we have 
$2l\leq \min\{n+i-2k,2k\}\leq n+i-2k$, which follows $2n+2i-2k-2l\geq n+i$. Hence, the rate $\rho$ 
of $\mathcal{Q}_{i}'$ is 
\begin{align*}
    \rho=\frac{n+i-k-l}{n+i}\geq \frac{1}{2}.
\end{align*}

By similar calculations, one can easily know that analogous conclusions also hold for $\mathcal{Q}_{2i}'$ ($i\geq 0$) 
and $\mathcal{Q}_1'$. Thus, summarizing the analysis as above, our propagation rules are new and can generate EAQECCs 
with rates greater than or equal to $\frac{1}{2}$ and positive net rates. Therefore, many of the derived EAQECCs can 
have a nice performance and can be further employed for constructing catalytic codes.

\subsection{New EAQECCs and MDS EAQECCs}

Using the Hermitian self-orthogonal MDS codes constructed in Section \ref{sec5}, new EAQECCs can be derived from 
Corollary \ref{coro.EAQECCS for Hermitian hulls} immediately. 
In accordance with Remark \ref{rem14.EA-quantum Singleton bound} $(2)$, we can easily verify that part of these  
new EAQECCs are MDS EAQECCs.

\begin{theorem}\label{th.GRS OUR NEW EAQECCS FROM TWO NEW HERMITIAN SELF-ORTHOGONAL CODES}
    Let $n=tn'$, where $n'\mid (q-1)$ and $1\leq t\leq \frac{\sqrt{q}-1}{n_1}$ with $n_1=\frac{n'}{\gcd(n',\sqrt{q}+1)}$. 
    Then for each $1\leq k\leq \lfloor \frac{\sqrt{q}+n'(t-1)}{\sqrt{q}+1} \rfloor$, the following statements hold. 
    \begin{itemize}
        \item [\rm (1)] There exists 
        an $[[n,k-l,n-k+1;n-k-l]]_{\sqrt{q}}$ EAQECC $\mathcal{Q}$ and 
        an $[[n,n-k-l,k+1;k-l]]_{\sqrt{q}}$ MDS EAQECC $\mathcal{Q}'$ for $0\leq l\leq k$.
        
        \item [\rm (2)] If $q>4$, there exists 
        an $[[n+i,k-l,\widetilde{d}_{i};n+i-k-l]]_{\sqrt{q}}$ EAQECC $\mathcal{Q}_{i}$ and 
        an $[[n+i,n+i-k-l,\widetilde{d}_{i}^{\bot_\frac{h}{2}};k-l]]_{\sqrt{q}}$ EAQECC $\mathcal{Q}_{i}'$ 
        for $i\geq 2$ and $0\leq l\leq k$, where $d\leq \widetilde{d}_{i}\leq n+i+1-k$.

    \end{itemize}
\end{theorem}

\begin{theorem}\label{th.GRS OUR NEW EAQECCS FROM TWO NEW HERMITIAN SELF-ORTHOGONAL CODES222}
    Let $q=p^h$ be odd. Assume that $(q-1)\mid {\rm lcm}(x_1,x_2)$ and $\gcd(x_2,q-1)\mid x_1(\sqrt{q}-1)$, where $x_1$ and $x_2$ are two integers. 
    Let $n=n_1n_2$, where $1\leq n_1\leq \frac{q-1}{\gcd(q-1,x_1)}$ and $n_2=\frac{q-1}{\gcd(q-1,x_2)}$. 
    Then for any $1\leq k\leq \lfloor \frac{\sqrt{q}+n-n_2}{\sqrt{q}+1} \rfloor$, the following statements hold. 
    \begin{itemize}
        \item [\rm (1)] There exists 
        an $[[n,k-l,n-k+1;n-k-l]]_{\sqrt{q}}$ EAQECC $\mathcal{Q}$ and 
        an $[[n,n-k-l,k+1;k-l]]_{\sqrt{q}}$ MDS EAQECC $\mathcal{Q}'$ for $0\leq l\leq k$.
        
        \item [\rm (2)] If $q>4$, there exists 
        an $[[n+i,k-l,\widetilde{d}_{i};n+i-k-l]]_{\sqrt{q}}$ EAQECC $\mathcal{Q}_{i}$ and 
        an $[[n+i,n+i-k-l,\widetilde{d}_{i}^{\bot_\frac{h}{2}};k-l]]_{\sqrt{q}}$ EAQECC $\mathcal{Q}_{i}'$ 
        for $i\geq 2$ and $0\leq l\leq k$, where $d\leq \widetilde{d}_{i}\leq n+i+1-k$.

    \end{itemize}
\end{theorem}

\begin{remark}\label{rem17.new EAQECCs} $\quad$
    \begin{enumerate}
        \item [\rm (1)] From Remarks \ref{rem11.th38.Hermitian self-orthogonal}-\ref{rem13.th40.Hermitian self-orthogonal222} and \ref{rem16.propagation comparisons}, 
        we can know that many EAQECCs and MDS EAQECCs obtained from Theorems 
        \ref{th.GRS OUR NEW EAQECCS FROM TWO NEW HERMITIAN SELF-ORTHOGONAL CODES} and 
        \ref{th.GRS OUR NEW EAQECCS FROM TWO NEW HERMITIAN SELF-ORTHOGONAL CODES222} should be new. 

        \item [\rm (2)] According to Corollaries \ref{coro.EAQECCS for Hermitian hulls} $(3)$ and $(4)$, one can also obtain many EAQECCs and MDS EAQECCs of 
        length $n+1$. In Theorems \ref{th.GRS OUR NEW EAQECCS FROM TWO NEW HERMITIAN SELF-ORTHOGONAL CODES} and 
        \ref{th.GRS OUR NEW EAQECCS FROM TWO NEW HERMITIAN SELF-ORTHOGONAL CODES222}, we omit these EAQECCs and MDS EAQECCs because MDS codes of these lengths with 
        Hermitian hulls of  arbitrary dimensions have been constructed in \cite{RefJ-5} and \cite{RefJ-12}.    

        \item [\rm (3)] Corollary \ref{coro.EAQECCS for Hermitian hulls} $(2)$ does not work when $q = 4$ due to specific restrictions. 
        Notably, due to availability of MDS EAQECCs, in Theorems \ref{th.GRS OUR NEW EAQECCS FROM TWO NEW HERMITIAN SELF-ORTHOGONAL CODES} and 
        \ref{th.GRS OUR NEW EAQECCS FROM TWO NEW HERMITIAN SELF-ORTHOGONAL CODES222}, we have divided Corollary \ref{coro.EAQECCS for Hermitian hulls} $(1)$ into two parts.

    \end{enumerate}
\end{remark}

In the following, we give some examples on new EAQECCs and MDS EAQECCs. 
For many $\sqrt{q}$-ary MDS EAQECCs, their lengths are greater than $\sqrt{q}+1$, and their minimum distances are greater than or equal to $\lceil \frac{\sqrt{q}}{2} \rceil$.

\begin{example}\label{example.QECCs}
    In this example, we give some $\sqrt[]{q}$-ary MDS EAQECCs with $c=0$ derived from Theorems 
    \ref{th.GRS OUR NEW EAQECCS FROM TWO NEW HERMITIAN SELF-ORTHOGONAL CODES} (1) and 
    \ref{th.GRS OUR NEW EAQECCS FROM TWO NEW HERMITIAN SELF-ORTHOGONAL CODES222} (1) in 
    Tables \ref{tab:QECCs1} and \ref{tab:QECCs2}, respectively. Readers can easily notice that these 
    $\sqrt{q}$-ary MDS EAQECCs have lengths larger than $\sqrt{q}+1$ and their minimum distances are 
    greater than or equal to $\lceil \frac{\sqrt{q}}{2} \rceil$. 
    Also, we recall that an EAQECC becomes a stabilizer QECC if $c=0$. 

    \begin{table}[H]
        \centering
        \caption{Some $9$-ary and $11$-ary MDS EAQECCs with $c=0$ from Theorem \ref{th.GRS OUR NEW EAQECCS FROM TWO NEW HERMITIAN SELF-ORTHOGONAL CODES} $(1)$}
        \label{tab:QECCs1}       
            \begin{tabular}{cc||cc}
                \hline
               $(q,n',t)$ & $\sqrt{q}$-ary MDS EAQECCs & $(q,n',t)$ & $\sqrt{q}$-ary MDS EAQECCs \\
                \hline
                $(3^4,5,8)$ & $[[40,32,5;0]]_{9}$ & $(3^4,10,5)$ & $[[50,42,5;0]]_{9}$ \\

                $(3^4,10,6)$ & $[[60,52,5;0]]_{9}$ & $(3^4,10,6)$ & $[[60,50,6;0]]_{9}$ \\

                $(3^4,10,7)$ & $[[70,62,5;0]]_{9}$ & $(3^4,10,7)$ & $[[70,60,6;0]]_{9}$ \\

                $(3^4,10,7)$ & $[[70,58,7;0]]_{9}$ & $(11^2,6,10)$ & $[[60,50,6;0]]_{11}$ \\
               
                $(11^2,12,6)$ & $[[72,62,6;0]]_{11}$ & $(11^2,12,7)$ & $[[84,74,6;0]]_{11}$ \\ 
                
                $(11^2,12,7)$ & $[[84,74,6;0]]_{11}$ & $(11^2,12,8)$ & $[[96,86,6;0]]_{11}$  \\ 
                
                $(11^2,12,8)$ & $[[96,84,7;0]]_{11}$ & $(11^2,12,8)$ & $[[96,82,8;0]]_{11}$ \\ 
                
                $(11^2,12,9)$ & $[[108,98,6;0]]_{11}$ & $(11^2,12,9)$ & $[[108,96,7;0]]_{11}$ \\ 
                
                $(11^2,12,9)$ & $[[108,94,8;0]]_{11}$ & $(11^2,12,9)$ & $[[108,92,9;0]]_{11}$ \\ 
                \hline
            \end{tabular}
    \end{table}

    \begin{table}[H]
        \centering
        \caption{Some $19$-ary and $23$-ary MDS EAQECCs with $c=0$ from Theorem \ref{th.GRS OUR NEW EAQECCS FROM TWO NEW HERMITIAN SELF-ORTHOGONAL CODES222} $(1)$}
        \label{tab:QECCs2}       
        \begin{tabular}{cc||cc}
            \hline
           $(q,x_1,x_2,n_1,n_2)$ & $\sqrt{q}$-ary MDS EAQECCs & $(q,x_1,x_2,n_1,n_2)$ & $\sqrt{q}$-ary MDS EAQECCs \\
            \hline
            $(19^2,40,9,6,40)$ & $[[240,222,10;0]]_{19}$ & $(19^2,40,9,6,40)$ & $[[240,220,11;0]]_{19}$ \\ 

            $(19^2,40,9,7,40)$ & $[[280,262,10;0]]_{19}$ & $(19^2,40,9,7,40)$ & $[[280,260,11;0]]_{19}$ \\ 
            $(19^2,40,9,7,40)$ & $[[280,258,12;0]]_{19}$ & $(19^2,40,9,7,40)$ & $[[280,256,13;0]]_{19}$ \\ 

            $(19^2,40,9,8,40)$ & $[[320,302,10;0]]_{19}$ & $(19^2,40,9,8,40)$ & $[[320,300,11;0]]_{19}$ \\ 
            $(19^2,40,9,8,40)$ & $[[320,298,12;0]]_{19}$ & $(19^2,40,9,8,40)$ & $[[320,296,13;0]]_{19}$ \\ 
            $(19^2,40,9,8,40)$ & $[[320,294,14;0]]_{19}$ & $(19^2,40,9,8,40)$ & $[[320,292,15;0]]_{19}$ \\

            $(23^2,48,11,7,48)$ & $[[336,314,12;0]]_{23}$ & $(23^2,48,11,7,48)$ & $[[336,312,13;0]]_{23}$ \\ 

            $(23^2,48,11,8,48)$ & $[[384,362,12;0]]_{23}$ & $(23^2,48,11,8,48)$ & $[[384,356,15;0]]_{23}$ \\ 

            $(23^2,48,11,9,48)$ & $[[432,402,16;0]]_{23}$ & $(23^2,48,11,9,48)$ & $[[432,400,17;0]]_{23}$ \\ 

            $(23^2,48,11,10,48)$ & $[[480,454,14;0]]_{23}$ & $(23^2,48,11,10,48)$ & $[[480,452,15;0]]_{23}$ \\ 
            $(23^2,48,11,10,48)$ & $[[480,450,16;0]]_{23}$ & $(23^2,48,11,10,48)$ & $[[480,448,17;0]]_{23}$ \\ 
            $(23^2,48,11,10,48)$ & $[[480,446,18;0]]_{23}$ & $(23^2,48,11,10,48)$ & $[[480,444,19;0]]_{23}$ \\

            \hline
        \end{tabular}
    \end{table}

\end{example}

\begin{example}\label{exam.EAQECCs3}
    As Remark \ref{rem16.propagation comparisons} demonstrates, our propagation rules are new, and thus, some EAQECCs with good and new 
    parameters can also be obtained from our propagation rules.  We give some concrete examples as follows. 

    \begin{enumerate}
        \item [\rm (1)] Consider the $[[4,0,2;0]]_2$ EAQECC constructed from the $[4,2,2]_4$ Hermitian self-orthogonal code listed in Example \ref{example2} $(1)$.  
        From Example \ref{example2} $(1)$, we can get a $[6,2,4]_4$ Hermitian self-orthogonal code and its Hermitian dual code has parameters $[6,4,2]_4$. 
        Applying Corollary \ref{coro.EAQECCS for Hermitian hulls} $(2)$, we can derive 
        $[[6,0,4;2]]_2$, $[[6,1,4;3]]_2$, $[[6,2,4;4]]_2$, $[[6,4,2;2]]_2$, $[[6,3,2;1]]_2$ and $[[6,2,2;0]]_2$ EAQECCs. 
        Compared with EAQECCs listed in \cite{RefJ-10}, we know that $[[6,2,2;0]]_2$ is optimal and $[[6,2,4;4]]_2$ is best-known. 
        In addition, one can easily check that optimal $[[8,4,2;0]]_2$, $[[10,6,2;0]]_2$, $[[12,8,2;0]]_2$, $[[14,10,2;0]]_2$, $\cdots$, EAQECCs can be similarly obtained. 

        \item [\rm (2)] Consider the $[[13,0,11;9]]_3$ EAQECC constructed from the $[13,2,11]_9$ Hermitian self-orthogonal code with the Herimitian dual code $[13,11,2]_9$ listed in Example \ref{example2} $(3)$.  
        From Example \ref{example2} $(3)$, we can get a $[15,2,13]_9$ Hermitian self-orthogonal code and its Hermitian dual code has parameters $[15,13,2]_9$. 
        Applying Corollary \ref{coro.EAQECCS for Hermitian hulls} $(1)$, we can derive 
        $[[13,0,11;9]]_3$, $[[13,1,11;10]]_3$, $[[13,2,11;11]]_3$, $[[13,11,2;2]]_3$, $[[13,10,2;1]]_3$, $[[13,9,2;0]]_3$,  
        $[[15,0,13;11]]_3$, $[[15,1,13;12]]_3$, $[[15,2,13;13]]_3$, $[[15,13,2;2]]_3$, $[[15,12,2;1]]_3$ and $[[15,11,2;0]]_3$ EAQECCs. 
        Compared with EAQECCs listed in \cite{RefJ-10}, we know that $[[13,0,11;9]]_3$, $[[13,2,11;11]]_3$, $[[15,0,13;11]]_3$ and $[[15,2,13;13]]_3$ are new,  
        $[[13,9,2;0]]_3$ and $[[15,11,2;0]]_3$ are optimal, and $[[13,1,11;10]]_3$ and $[[15,1,13;12]]_3$ are best-known. 
        For a more visual presentation, we only list new $3$-ary EAQECCs in Table \ref{tab:3-ary EAQECCs}.

    \begin{table}[H]
        \centering
        \caption{New $3$-ary $[[n,k,d;c]]_3$ EAQECCs from Example \ref{exam.EAQECCs3} $(2)$}
        \label{tab:3-ary EAQECCs}       
            \begin{tabular}{c|c||c|c}
                \hline
                \tabincell{c}{The known EAQECCs in \cite{RefJ-10}}  & \tabincell{c}{New EAQECCs} & \tabincell{c}{The known EAQECCs in \cite{RefJ-10}}  & \tabincell{c}{New EAQECCs}  \\
                \hline
                \tabincell{l}{$[[13,0,6;0]]_3$\\ $[[13,0,7;3]]_3$\\ $[[13,0,8;5]]_3$\\ $[[13,0,10;9]]_3$\\ $[[13,0,13;11]]_3$\\ $[[13,0,14;13]]_3$} & $\mathbf{[[13,0,11;9]]_3}$ & 
                \tabincell{l}{$[[13,2,5;0]]_3$\\ $[[13,2,7;2]]_3$\\ $[[13,2,8;6]]_3$\\ $[[13,2,9;8]]_3$\\ $[[13,2,10;9]]_3$} & $\mathbf{[[13,2,11;11]]_3}$  \\ \hline
    
                \tabincell{l}{$[[15,0,6;0]]_3$\\ $[[15,0,7;1]]_3$\\ $[[15,0,8;3]]_3$\\ $[[15,0,9;7]]_3$\\ $[[15,0,12;9]]_3$\\ $[[15,0,15;13]]_3$\\ $[[15,0,16;15]]_3$} & $\mathbf{[[15,0,13;11]]_3}$ & 
                \tabincell{l}{$[[15,2,5;0]]_3$\\ $[[13,2,6;1]]_3$\\ $[[15,2,9;2]]_3$\\ $[[15,2,10;8]]_3$\\ $[[15,2,11;10]]_3$\\ $[[15,2,12;11]]_3$} & $\mathbf{[[15,2,13;13]]_3}$  \\
    
                \hline
            \end{tabular}
    \end{table}
    \end{enumerate}
\end{example}

\begin{remark}\label{rem}
    The facts in Example \ref{exam.EAQECCs3} show intuitively that once an initial self-orthogonal code or an EAQECC with good parameters is given, 
    many good EAQECCs can be easily obtain by applying our propagation rules. 
    Naturally, taking similar calculations and more searches by the Magma software package \cite{BCP1997}, we believe that more new and good EAQECCs can be obtained.
\end{remark}

\section{Summary and concluding remarks}\label{sec7}

As a generalization of the Euclidean inner product and the Hermitian inner product, the Galois inner product has drawn 
a tremendous amount of attention from researchers recently. Due to the applications of Euclidean and Hermitian self-orthogonal 
codes in constructing classical QECCs and EAQECCs, one would also like to explore possible similar applications on Galois 
self-orthogonal codes. Some outstanding studies have given positive answers to this question. Therefore, the construction 
of Galois self-orthogonal codes over finite fields and the construction of codes with Galois hulls of arbitrary dimensions have 
become two very meaningful topics.  

In this paper, we first present some general properties of the dimension of the Galois hulls of linear codes. 
Then we systematically investigate how to construct Galois self-orthogonal codes of larger length from known Galois self-orthogonal codes. 
Automatically, codes of larger length with Galois hulls of arbitrary dimensions are derived. 
Two new families of Hermitian self-orthogonal MDS codes are constructed and are used to construct new EAQECCs and MDS EAQECCs. 
In particular, in our constructions, optimal EAQECCs, best-known EAQECCs, and EAQECCs with rates $\rho\geq \frac{1}{2}$ and net rates $\bar{\rho}>0$ exist. 
As a further application, catalytic codes can be obtained. Some interesting self-orthogonal codes and $m$-MDS codes are also given as examples.

We remark two interesting problems for the future work. 
The first one is to narrow down the range of the minimum distances of self-orthogonal codes of larger length in this paper. 
And the second one is to classify these Galois self-orthogonal codes.  

\section*{Acknowledgments}
The authors thank Prof. Markus Grassl for insightful discussions. His comments on an earlier draft of this work 
led us to further compare known propagation rules with ours and to consider the application of the Galois 
inner product to the construction of EAQECCs.
This research is supported by the National Natural Science Foundation of China (Nos.U21A20428 and 12171134). 
All the examples in this paper were computed with the Magma software package.

\section*{}

\end{sloppypar}
\end{document}